\documentclass[a4paper,USenglish, 11pt]{article}
\usepackage[T1]{fontenc}
\usepackage{amsthm,amsmath,amssymb,amsfonts,authblk,hyperref,setspace}
\usepackage{graphicx}
\usepackage{color}

\newtheorem{theorem}{Theorem}[section]
\newtheorem{conjecture}{Conjecture}

\newtheorem{lemma}[theorem]{Lemma}
\newtheorem{claim}[theorem]{Claim}
\newtheorem{proposition}[theorem]{Proposition}

\newtheorem{corollary}[theorem]{Corollary}

\newtheorem{definition}{Definition}

\newtheorem*{remark}{Remark}

\usepackage{fullpage}
\newcommand{\comment}[1]{}

\newcommand{\Read}{\textsc{READ}}
\newcommand{\Write}{\textsc{WRITE}}
\newcommand{\ord}{\texttt{ord}}

\newcommand{\lcm}{\texttt{lcm}}

\newcommand{\Z}{\mathbb{Z}}
\newcommand{\N}{\mathbb{N}}
\newcommand{\D}{\mathcal{D}}

\newcommand{\<}{\langle}
\renewcommand{\>}{\rangle}

\newcommand{\om}{o}

\begin{document}



\title{Optimal Quasi-Gray Codes: The Alphabet Matters\footnote{The research leading to these results has received funding from the European Research Council under the European Union's Seventh Framework Programme (FP/2007-2013)/ERC Grant Agreement no. 616787. The third author was also partially supported by the Center of Excellence CE-ITI under the grant P202/12/G061 of GA \v{C}R.}}

\author[1]{Diptarka Chakraborty\thanks{diptarka@iuuk.mff.cuni.cz}}
\author[2]{Debarati Das\thanks{debaratix710@gmail.com}}
\author[3]{Michal Kouck{\'{y}}\thanks{koucky@iuuk.mff.cuni.cz}}
\author[4]{Nitin Saurabh\thanks{nsaurabh@mpi-inf.mpg.de}}
\affil[1,2,3]{Computer Science Institute of Charles University,
Malostransk{\'e}  n{\'a}m\v{e}st{\'\i} 25,
118 00 Praha 1, Czech Republic}
\affil[4]{Max-Planck-Institut f\"{u}r Informatik, Saarbr\"{u}cken, Germany}

\maketitle
\begin{abstract}
A quasi-Gray code of dimension $n$ and length $\ell$ over an alphabet $\Sigma$ is
a sequence of distinct words $w_1,w_2,\dots,w_\ell$ from $\Sigma^n$ such that
any two consecutive words differ in at most $c$ coordinates, for some fixed constant $c>0$.
In this paper we are interested in the read and write complexity of quasi-Gray codes
in the bit-probe model, where we measure the number of symbols read and written in order
to transform any word $w_i$ into its successor $w_{i+1}$. 

We present construction of quasi-Gray codes of dimension $n$ and length $3^n$ over
the ternary alphabet $\{0,1,2\}$ with worst-case read complexity $O(\log n)$ and write complexity $2$.
This generalizes to arbitrary odd-size alphabets. For the binary alphabet, we present
quasi-Gray codes of dimension $n$ and length at least $2^n - 20n$ with worst-case read complexity $6+\log n$
and write complexity $2$. This complements a recent result by Raskin [Raskin '17] who shows that any quasi-Gray code
over binary alphabet of length $2^n$ has read complexity $\Omega(n)$.

Our results significantly improve on previously known constructions and for the odd-size
alphabets we break the $\Omega(n)$ worst-case barrier for space-optimal (non-redundant)
quasi-Gray codes with constant number of writes. We obtain our results via a novel application
of algebraic tools together with the principles of catalytic computation [Buhrman et al. '14, Ben-Or and Cleve '92, Barrington '89, Coppersmith and Grossman '75].
\end{abstract}
\newpage
\section{Introduction}
One of the fundamental problems in the domain of algorithm design is
to list down all the objects belonging to a certain combinatorial class.
Researchers are interested in efficient generation of a list such that
an element in the list can be obtained by a small amount of change
to the element that precedes it.
One of the classic examples is the binary \emph{Gray code} introduced
by Gray~\cite{Gray53}, initially used in pulse code communication.
The original idea of a Gray code was to list down all binary strings
of length $n$, i.e, all the elements of $\Z_2^n$, 
such that any two successive strings differ by exactly one bit.
The idea was later generalized for other combinatorial classes (e.g. see~\cite{NW78, Knuth11}).
Gray codes have found applications in a wide variety of areas,
such as information storage and retrieval~\cite{CCC92}, processor allocation~\cite{CS90}, computing the permanent~\cite{NW78}, circuit testing~\cite{RC81},
data compression~\cite{Rich86}, graphics and image processing~\cite{ASD90},
signal encoding~\cite{Ludman81}, modulation schemes for flash memories~\cite{JMSB09,GLSB11,YS12} and many more. 
Interested reader may refer to an excellent survey by
Savage~\cite{Sav97} for a comprehensive treatment
on this subject.

In this paper we study the construction of
Gray codes over $\Z_m^n$ for any $m \in \N$.
Originally, Gray codes were meant to list down all the elements from its domain but later
studies (e.g.~\cite{Fred78, RM10, BCJMMS10, BGPS14}) focused  on the generalization where we list $\ell$ distinct elements from the domain,
each two consecutive elements differing in one position. We refer to such codes as Gray codes of length $\ell$~\cite{Fred78}.
When the code lists all the elements from its domain it is referred to as {\em space-optimal}.
It is often required
that the last and the first strings appearing in the list
also differ in one position. Such codes are called \emph{cyclic Gray codes}.
Throughout this paper we consider only cyclic Gray codes and
we refer to them simply as Gray codes. Researchers also study codes
where two successive strings differ in at most $c$ positions, 
for some fixed constant $c>0$, instead of differing in exactly one position. Such codes are called 
\emph{quasi-Gray codes}~\cite{BCJMMS10}\footnote{Readers may note that
the definition of quasi-Gray code given in~\cite{Fred78} was different.
The code referred as quasi-Gray code by Fredman~\cite{Fred78} is called Gray code of length $\ell$
where $\ell < m^n$, in our notation.} or $c$-Gray codes.

We study the problem of constructing quasi-Gray codes over
$\Z_m^n$ in the cell probe model~\cite{Yao81a},
where each cell stores an element from $\Z_m$.
The efficiency of a construction is measured using three parameters.
First, we want the length of a quasi-Gray code to be as large as possible. Ideally, we want space-optimal codes. Second,
we want to minimize the number of coordinates of the input string
the algorithm reads in order to generate the next (or, previous) string in the code.
Finally, we also want the number of cells written in order
to generate the successor (or, predecessor) string 
to be as small as possible.
Since our focus is on quasi-Gray codes, the number of writes will always
be bounded by a universal constant. We are interested in the worst-case behavior and we use \emph{decision assignment trees} (DAT) of Fredman~\cite{Fred78} to measure these complexities. 

The second requirement of the above is motivated from the study of {\em loopless generation} of combinatorial objects. In the loopless generation we are required to generate the next string from the code in constant time.
Different loopless algorithms to generate Gray codes are known in
the literature~\cite{Ehr73, BER76, Knuth11}. However, those algorithms use extra memory cells in addition to the space required for
the input string which makes it impossible to get a space-optimal code from them.
More specifically, our goal is to design a decision assignment tree on $n$ variables to generate a code over the domain $\mathbb{Z}_m^n$.
If we allow extra memory cells (as in the case of loopless algorithms) then the corresponding DAT will be on $n+b$ variables, where $b$
is the number of extra memory cells used.

Although there are known quasi-Gray codes with logarithmic read complexity and constant write complexity 
\cite{Fred78, RM10, BCJMMS10, BGPS14}, none of these constructions is space-optimal. The best result misses at least $2^{n-t}$
strings from the domain when having read complexity $t+O(\log n)$ \cite{BGPS14}. Despite of an extensive research 
under many names, e.g., construction of Gray codes~\cite{Fred78, NRR13, DGK17, GM17}, dynamic language membership problem~\cite{FMS97},
efficient representation of integers~\cite{RM10, BGPS14}, so far we do not have any quasi-Gray code of length $2^n - 2^{\epsilon n}$, for some constant $\epsilon < 1$,
with worst-case read complexity $o(n)$ and write complexity $o(n)$. The best worst-case read complexity for {\em space-optimal} Gray code is $n-1$ \cite{Fren16}. Recently, Raskin~\cite{Raskin17} showed that any space-optimal quasi-Gray code over the domain $\mathbb{Z}_2^n$ must have read complexity $\Omega(n)$.
This lower bound is true even if we allow non-constant write complexity.
It is worth noting that this result can be extended to 
the domain $\mathbb{Z}_m^n$ when $m$ is even.

In this paper we show that such lower bound does not hold for quasi-Gray
codes over $\Z_m^n$, when $m$ is odd. In particular, we construct space-optimal quasi-Gray codes over $\{0,1,2\}^n$
with read complexity $4 \log_3 n$ and write complexity $2$.  As a consequence we get an exponential separation between the read complexity of space-optimal quasi-Gray code over $\Z_2^n$ and that over $\Z_3^n$.

\begin{theorem}
\label{thm:odd-counter}
Let $m \in \N$ be odd and $n\in \N$ be such that $n \geq 15$. 
Then, there is a space-optimal
quasi-Gray code $C$ over $\Z_m^n$ for which, the two functions $next(C,w)$ and $prev(C,w)$ can be implemented by inspecting at most $4\log_m n$ cells while writing only $2$ cells.
\end{theorem}


In the statement of the above theorem, $next(C,w)$
denotes the element appearing after $w$ in the cyclic sequence of the code $C$, and analogously, $prev(C,w)$ denotes the preceding element. Using the argument as in~\cite{Fred78, NRR13} it is easy to see a lower bound of $\Omega\left(\log_m n \right)$ on the read complexity when the domain is $\Z_m^n$. Hence our result is optimal up to some small constant factor.

Raskin shows $\Omega(n)$ lower bound on the read complexity of space-optimal binary quasi-Gray codes. The existence of binary quasi-Gray codes with sub-linear read complexity of length $2^n - 2^{\epsilon n}$, for some constant $\epsilon < 1$, was open. Using a different technique than that used in the proof of Theorem~\ref{thm:odd-counter} we get a quasi-Gray code
over the binary alphabet which enumerates all but $O(n)$ many strings. This result generalizes to the domain $\Z_q^n$ for any prime power $q$.

\begin{theorem}
  \label{thm:binary-counter}
  Let $n \geq 15$ be any natural number.  
  Then, there is a quasi-Gray
code $C$ of length at least $2^n-20n$ over $\Z_2^n$, such that the two functions 
$next(C,w)$ and $prev(C,w)$ can be implemented by inspecting at most $6+\log n$ cells
while writing only $2$ cells.
\end{theorem}


We remark that the points that are missing from $C$ in the above theorem
are all of the form $\{0,1\}^{O(\log n)} 0^{n-O(\log n)}$.

If we are allowed to read and write constant fraction of $n$ bits 
then Theorem~\ref{thm:binary-counter} can be adapted to get a quasi-Gray code
of length $2^n - O(1)$ (see Section~\ref{sec:binarycounter}). In this way we get a trade-off
between length of the quasi-Gray code and the number of bits read in the worst-case.
All of our constructions can be made uniform (see the remark after Corollary \ref{cor:counter-prime-power}).

Using the Chinese Remainder Theorem (cf.~\cite{DPS96}), we also develop a technique
that allows us to compose Gray codes over various domains.
Hence, from quasi-Gray codes over domains $\Z_{m_1}^n, \Z_{m_2}^n, \dots, \Z_{m_k}^n$, where $m_i$'s are pairwise co-prime,
we can construct quasi-Gray codes over $\Z_m^{n'}$, where $m=m_1\cdot m_2\cdots m_k$.
Using this technique on our main results, we get a quasi-Gray code over
$\Z_m^n$ for {\em any} $m \in \N$
that misses only $O(n\cdot\om^n)$ strings
where $m=2^\ell\om$ for an odd $\om$,
while achieving the read complexity similar to that stated in
Theorem~\ref{thm:odd-counter}.
It is worth mentioning that if we get a space-optimal
quasi-Gray code over the binary alphabet with non-trivial savings in
read complexity, then we will have a space-optimal quasi-Gray code
over the strings of alphabet $\Z_m$ for any $m \in \N$
with similar savings.

The technique by which we construct our quasi-Gray codes relies heavily on simple algebra
which is a substantial departure from previous mostly combinatorial constructions. We view Gray codes
as permutations on $\Z_m^{n}$ and we decompose them into $k$ simpler permutations on $\Z_m^{n}$, each
being computable with read complexity $3$ and write complexity $1$. Then we apply a different composition theorem, than mentioned
above, to obtain space-optimal quasi-Gray codes on $\Z_m^{n'}$, $n'={n+\log k}$, with read complexity $O(1)+\log k$ and write complexity 2.
The main issue is the decomposition of permutations into few simple permutations. This is achieved by techniques
of catalytic computation \cite{BCKLS14} going back to the work of  Coppersmith and Grossman \cite{CG75, Barr89, BC92}.

It follows from the work of Coppersmith and Grossman \cite{CG75} that our technique is incapable 
of designing a space-optimal quasi-Gray code on $\Z_2^{n'}$ as any such code represents an {\em odd} permutation.
The tools we use give inherently only {\em even} permutations. However, we can construct quasi-Gray codes
from cycles of length $2^{n}-1$ on $\Z_2^{n}$  as they are even permutations. Indeed, that is what we
do for our Theorem~\ref{thm:binary-counter}. We note that any efficiently computable odd permutation on $\Z_2^{n}$,
with say read complexity $(1-\epsilon)n$ and write complexity $O(1)$, could be used together with our
technique to construct a space-optimal quasi-Gray code on $\Z_2^{n'}$ with read complexity at most $(1-\epsilon')n'$
and constant write complexity. This would represent a major progress on space-optimal Gray codes.
(We would compose the odd permutation with some even permutation to 
obtain a full cycle on $\Z_2^{n}$. The size of the decomposition of the even permutation into simpler permutations
would govern the read complexity of the resulting quasi-Gray code.)

Interestingly, Raskin's result relies on showing that a decision assignment tree of sub-linear read complexity
must compute an even permutation.


\subsection{Related works}
The construction of Gray codes is central to the design of algorithms for many combinatorial problems~\cite{Sav97}. Frank Gray~\cite{Gray53} first came up with a construction of Gray code over binary strings of length $n$, where to generate the successor or predecessor strings one needs to read $n$ bits in the worst-case. The type of code described in~\cite{Gray53} is known as \emph{binary reflected Gray code}. Later Bose \emph{et al.}~\cite{BCJMMS10} provided a different type of Gray code construction, namely \emph{recursive partition Gray code} which attains $O(\log n)$ average case read complexity while having the same worst-case read requirements. The read complexity we referred here is in the bit-probe model. It is easy to observe that any space-optimal binary Gray code must read $\log n + 1$ bits in the worst-case~\cite{Fred78, NRR13, Fren16}. Recently, this lower bound was improved to $n/2$ in~\cite{Raskin17}. An upper bound of even $n-1$ was not known until very recently~\cite{Fren16}. This is also the best known so far. 

Fredman~\cite{Fred78} extended the definition of Gray codes by considering codes that may not enumerate all the strings (though presented in a slightly different way in~\cite{Fred78}) and also introduced the notion of \emph{decision assignment tree} (DAT) to study the complexity of any code in the bit-probe model. He provided a construction that generates a Gray code of length $2^{c\cdot n}$ for some constant $c < 1$ while reducing the worst-case bit-read to $O(\log n)$. Using the idea of Lucal's modified reflected binary code~\cite{Luc59}, Munro and Rahman~\cite{RM10} got a code of length $2^{n-1}$ with worst-case read complexity only $4+\log n$. However in their code two successive strings differ by $4$ coordinates in the worst-case, instead of just one and we refer to such codes as quasi-Gray codes following the nomenclature used in~\cite{BCJMMS10}. Brodal et al. \cite{BGPS14} extended the results of~\cite{RM10} by constructing a quasi-Gray code of length $2^n-2^{n-t}$ for arbitrary $1\le t \le n- \log n -1$, that has $t +3 + \log n$ bits ($t +2 + \log n$ bits) worst-case read complexity and any two successive strings in the code differ by at most $2$ bits ($3$ bits).

In contrast to the Gray codes over binary alphabets, Gray codes over non-binary alphabets received much less attention.
The construction
of binary reflected Gray code was generalized to the alphabet
$\Z_m$ for any $m\in \N$
in~\cite{Flores56, Cohn63, JWW80, Rich86, Knuth11, HR16}. However, each of those
constructions reads $n$ coordinates in the worst-case to generate the
next element. 
As mentioned before, we measure the read complexity
in the well studied cell probe model~\cite{Yao81a} where we assume
that each cell stores an element of $\Z_m$.
The argument of Fredman in~\cite{Fred78} implies a lower bound of $\Omega(\log_m n)$ on
the read complexity of quasi-Gray code on $\Z_m^n$.
To the best of our knowledge, for non-binary alphabets, there is nothing known similar
to the results of Munro and Rahman or Brodal et al. ~\cite{RM10, BGPS14}.
We summarize the previous results along with ours in
Table~\ref{table:comparison}.

\begin{table}[]
\centering
\scalebox{1.00}{
\begin{tabular}{|c|c|l|l|c|}
\hline
Reference    &Value of $m$    & length             & Worst-case cell read     & Worst-case cell write \\ \hline
\cite{Gray53}   & $2$          & $2^n$    &  $n$                  & $1$                               \\ \hline
\cite{Fred78}   & $2$          & $2^{\Theta(n)}$    &  $O(\log n)$                  & $O(1)$                               \\ \hline
\cite{FMS97}   & $2$          & $\Theta(2^n/n)$    &  $\log n + 1$                  & $\log n + 1$                               \\ \hline
\cite{RM10}   & $2$          & $2^{n-1}$    &  $\log n + 4$                  & $4$                               \\ \hline
\cite{BCJMMS10}   & $2$          & $2^n-O(2^n/n^t)$    &  $O(t \log n)$                  & $3$                               \\ \hline
\cite{BGPS14}   & $2$          & $2^n-2^{n-t}$    &  $\log n + t +3$                  & $2$                               \\ \hline
\cite{BGPS14}   & $2$          & $2^n-2^{n-t}$    &  $\log n + t +2$                  & $3$                               \\ \hline
\cite{Fren16}   & $2$          & $2^n$    &  $n-1$                  & $1$                               \\ \hline
Theorem~\ref{thm:binary-counter}   & $2$          & $2^n-O(n)$    &  $\log n + 4$                  & $2$                               \\ \hline
\cite{Cohn63}   & any $m$          & $m^n$    &  $n$                  & $1$                               \\ \hline
Theorem~\ref{thm:odd-counter}   & any odd $m$          & $m^n$    &  $4\log_m n + 3$                  & $2$                               \\ \hline

\end{tabular}}
\caption{Taxonomy of construction of Gray/quasi-Gray codes over $\Z_m^n$}
\label{table:comparison}
\end{table}

Additionally, many variants of Gray codes have been studied in the literature.
A particular one that has garnered a lot of attention in the past 30 years
is the well-known \emph{middle levels conjecture}. See \cite{Mutze16,MN15,MN17, GM17}, and the references therein. It has been established only recently
\cite{Mutze16}. The conjecture says that there exists a Hamiltonian cycle in the graph
induced by the vertices on levels $n$ and $n+1$ of the hypercube graph
in $2n+1$ dimensions. In other words, there exists a Gray code on
the middle levels. M\"{u}tze et al.~\cite{MN15,MN17} studied
the question of efficiently enumerating such a Gray code in the word RAM model.
They \cite{MN17} gave an algorithm to enumerate a Gray code in the middle levels
that requires $O(n)$ space and
\emph{on average} takes $O(1)$ time to generate the next vertex. 
In this paper we consider the bit-probe model, and Gray codes over the complete
hypercube. It would be interesting to know
whether our technique can be applied for the middle level Gray codes.

\subsection{Our technique}
\label{sec:our-technique}

Our construction of Gray codes relies heavily on the notion of $s$-functions defined by Coppersmith and Grossman \cite{CG75}. An $s$-function
is a permutation $\tau$ on $\Z_m^n$ defined by a function $f:\Z_m^s \rightarrow \Z_m$ and an $(s+1)$-tuple of indices $i_1,i_2,\dots,i_s,j \in [n]$
such that $\tau(\<x_1,x_2,\dots,x_n\>)=(\<x_1,x_2,\dots,x_{j-1},x_j+f(x_{i_1},\dots,x_{i_s}),x_{j+1},\dots,x_n\>)$, where the addition is inside $\Z_m$.
Each $s$-function can be computed by some decision assignment tree that given a vector $x=\<x_1,x_2,\dots,$ $x_n\>$, inspects $s+1$ coordinates of $x$ and then
it writes into a single coordinate of $x$.

A counter $C$ (quasi-Gray code) on $\Z_m^n$ can be thought of as a permutation on $\Z_m^n$. Our goal is to construct some permutation $\alpha$ on $\Z_m^n$ that can be written
as a composition of $2$-functions $\alpha_1,\dots,\alpha_k$, i.e., $\alpha=\alpha_k \circ \alpha_{k-1} \circ \cdots \circ \alpha_1$.

Given such a decomposition, we can build another counter $C'$ on $\Z_m^{r+n}$, where $r=\lceil \log_m k \rceil$, for which the function $next(C',x)$
operates as follows. The first $r$-coordinates of $x$ serve as an {\em instruction pointer} $i\in [m^r]$ that determines which $\alpha_i$ should be
executed on the remaining $n$ coordinates of $x$. Hence, based on the current value $i$ of the $r$ coordinates, we perform $\alpha_i$ on the remaining
coordinates and then we update the value of $i$ to $i+1$. (For $i>k$ we can execute the identity permutation which does nothing.) 

We can use known Gray codes on $\Z_m^r$ to represent the instruction pointer so that when incrementing $i$ we only need to write into one
of the coordinates. This gives a counter $C'$ which can be computed by a decision assignment tree that reads $r+3$ coordinates and writes into $2$ coordinates of $x$.
(A similar composition technique is implicit in Brodal et al.~\cite{BGPS14}.)
If $C$ is of length $\ell=m^n -t$, then $C'$ is of length $m^{n+r}-tm^r$. In particular, if $C$ is space-optimal then so is $C'$.

Hence, we reduce the problem of constructing $2$-Gray codes to the problem of designing large cycles in $\Z_m^n$ that can be decomposed into $2$-functions.
Coppersmith and Grossman \cite{CG75} studied precisely the question of, which permutations on $\Z_2^n$ can be written as a composition of $2$-functions.
They show that a permutation on $\Z_2^n$ can be written as a composition of $2$-functions if and only if the permutation is even. Since $\Z_2^n$ is of even size,
a cycle of length $2^n$ on $\Z_2^n$ is an odd permutation and thus it cannot be represented as a composition of $2$-functions. However, their result
also implies that a cycle of length $2^n-1$ on $\Z_2^n$ {\em can} be decomposed into $2$-functions. 

We want to use the counter composition technique described above in connection with a cycle of length $2^n-1$.
To maximize the length of the cycle $C'$ in $\Z_2^{n+r}$, we need to minimize $k$,
the number of $2$-functions in the decomposition. By a simple counting argument, most cycles of length $2^n-1$ on $\Z_2^n$ require $k$
to be exponentially large in $n$. This is too large for our purposes. Luckily, there are cycles of length  $2^n-1$ on $\Z_2^n$ 
that can be decomposed into polynomially many $2$-functions, and we obtain such cycles from linear transformations.

There are linear transformations $\Z_2^n \rightarrow \Z_2^n$ which define a cycle on $\Z_2^n$ of length  $2^n-1$. 
For example, the matrix corresponding to the multiplication by a fixed generator of the multiplicative group  $\mathbb{F}_{2^n}^{\ast}$ 
of the Galois field $GF[2^n]$
is such a matrix. Such matrices are full rank and they can be decomposed into $O(n^2)$
elementary matrices, each corresponding to a $2$-function. Moreover, there are matrices derived from primitive polynomials
that can be decomposed into at most $4n$ elementary matrices.\footnote{Primitive polynomials were previously also used in a similar problem, namely to construct
shift-register sequences (see e.g. \cite{Knuth11}).}
 We use them to get a counter on $\Z_2^{n'}$ of length at least $2^{n'} - 20n'$
whose successor and predecessor functions are computable by decision assignment trees of read complexity $\le 6+\log n'$ and write complexity $2$.
Such counter represents $2$-Gray code of the prescribed length.
For any prime $q$, the same construction yields $2$-Gray codes of length at least $q^{n'}-5q^2 n'$ with decision assignment trees of read complexity $\le 6+\log_q n'$ and write complexity $2$.

The results of Coppersmith and Grossman \cite{CG75} can be generalized to $\Z_m^n$ as stated in Richard Cleve's thesis \cite{Cl89}.\footnote{Unfortunately, there
is no written record of the proof.} For odd $m$, if a permutation on $\Z_m^n$ is even then it can be decomposed into $2$-functions. Since $m^n$ is odd,
a cycle of length $m^n$ on $\Z_m^n$ is an even permutation and so it can be decomposed into $2$-functions. If the number $k$ of those functions is small,
so the $\log_m k$ is small, we get the sought after counter with small read complexity. 
However, for most cycles of length $m^n$ on $\Z_m^n$, $k$ is exponential in $n$.

We show though, that there is a cycle $\alpha$ of length $m^n$ on $\Z_m^n$ that can be decomposed into $O(n^3)$ $2$-functions. This in turn gives space-optimal
$2$-Gray codes on $\Z_m^{n'}$ with decision assignment trees of read complexity $O(\log_m n'$) and write complexity $2$.

We obtain the cycle $\alpha$ and its decomposition in two steps. First, for $i\in [n]$, we consider the permutation $\alpha_i$ on $\Z_m^n$
which maps each element $0^{i-1}ay$ onto $0^{i-1}(a+1)y$, for $a\in \Z_m$ and $y\in \Z_m^{n-i}$, while other elements are mapped to themselves.
Hence, $\alpha_i$ is a product of $m^{n-i}$ disjoint cycles of length $m$. 
We show that $\alpha = \alpha_n \circ \alpha_{n-1} \circ \cdots \circ \alpha_1$ is a cycle of length $m^n$. 
In the next step we decompose each $\alpha_i$ into $O(n^2)$ $2$-functions. 

For $i\le n-2$, we can decompose $\alpha_i$ using the technique
of Ben-Or and Cleve \cite{BC92} and its refinement in the form of catalytic computation of Buhrman et al. \cite{BCKLS14}.
We can think of $x\in \Z_m^n$ as content of $n$ memory registers, where $x_1,\dots,x_{i-1}$ are the input registers, $x_i$ is the output register,
and $x_{i+1},\dots,x_n$ are the working registers. The catalytic computation technique gives a program consisting of $O(n^2)$ instructions,
each being equivalent to a $2$-function, which performs the desired adjustment of $x_i$ based on the values of $x_1,\dots,x_{i-1}$
without changing the ultimate values of the other registers. (We need to increment $x_i$ iff $x_1,\dots,x_{i-1}$ are all zero.)
This program directly gives the desired decomposition of $\alpha_i$, for $i\le n-2$. (Our proof in Section \ref{sec:oddcounter} uses the language
of permutations.)

The technique of catalytic computation fails for $\alpha_{n-1}$ and $\alpha_n$ as the program needs at least two working registers to operate.
Hence, for  $\alpha_{n-1}$ and $\alpha_n$ we have to develop entirely different technique. This is not trivial and quite technical but it is nevertheless possible,
thanks to the specific structure of  $\alpha_{n-1}$ and $\alpha_n$.

\subsection*{Organization of the paper}
In Section~\ref{sec:prelim} we define the notion of counter, Gray code and our computational model, namely decision assignment tree and also provide certain known results regarding the construction of Gray codes. Then in Section~\ref{sec:chinese} we describe how to combine counters over smaller alphabets to get another counter over larger alphabet, by introducing the Chinese Remainder Theorem for counters. Next we provide some basic facts about the permutation group and the underlying structure behind all of our constructions of quasi-Gray codes. 
We devote Section~\ref{sec:binarycounter} to construction of a quasi-Gray code over binary alphabet that misses only a few words, by using full rank linear transformation. 
In Section~\ref{sec:oddcounter} we construct a space-optimal quasi-Gray code over any odd-size alphabet. 
Finally in Section~\ref{sec:LB-hierarchy} we rule out the existence of certain kind of space-optimal binary counters.

\section{Preliminaries}
\label{sec:prelim}

In the rest of the paper we only present constructions of the successor function $next(C,w)$ for our codes.
Since all the operations in those constructions are readily invertible, the same arguments also give the predecessor function $prev(C,w)$.

\paragraph*{Notations:} We use the standard notions of groups and
fields, and mostly we will use only elementary facts about them (see \cite{DF04,LN96} for background.). By $\Z_m$ we mean
the set of integers modulo $m$, i.e.,
$\Z_m := \Z/m \Z$. Throughout this paper
whenever we use addition and multiplication operation between
two elements of $\Z_m$, then we mean the operations within $\Z_m$ that is  modulo $m$.
For any $m \in \N$, we let $[m]$ denote
the set $\{1,2,\ldots,m\}$. Unless stated otherwise explicitly,
all the logarithms we consider throughout this paper are based $2$.

Now we define the notion of counters used in this paper.

\begin{definition}[Counter]
A \emph{counter} of length $\ell $ over a domain $\D$ is any cyclic sequence $C=(w_1,\ldots ,w_\ell)$ such that
$w_1, \dots, w_\ell$ are distinct elements of $\D$. With the counter $C$ we associate two functions  $next(C,w)$ and $prev(C,w)$
that give the successor and predecessor element of $w$ in $C$, that is for  $i\in [\ell]$, 
$next(C,w_i)=w_j$ where $j-i=1 \bmod \ell$, and $prev(C,w_i)=w_k$ where $i-k=1 \bmod \ell$.
If $\ell =|\D|$, we call the counter a \emph{space-optimal counter}.
\end{definition}

Often elements in the underlying domain $\D$ have some
``structure'' to them. In such cases, it is desirable to have
a counter such
that consecutive elements in the sequence differ by a ``small'' change
in the ``structure''. We make this concrete in
the following definition. 

\begin{definition}[Gray Code]
\label{def:graycode}
Let $\D_1,\dots,\D_n$ be finite sets.
A \emph{Gray code} of length $\ell $ over the domain
$\D=\D_1\times \cdots \times \D_n$ is
a counter $(w_1,\dots,w_\ell)$ of length $\ell $ over $\D$ such that any
two consecutive strings $w_i$ and $w_j$, $j-i=1 \bmod \ell$, differ in exactly one coordinate when
viewed as an $n$-tuple. More generally, if for some constant $c\ge 1$, any
two consecutive strings $w_i$ and $w_j$, $j-i=1 \bmod \ell$, differ in at most $c$ coordinates such a counter is called a \emph{$c$-Gray Code}.
\end{definition}

By a quasi-Gray code we mean $c$-Gray code for some unspecified fixed $c>0$.
In the literature sometimes people do not place any restriction on the relationship between $w_\ell$ and $w_1$
and they refer to such a sequence a (quasi)-Gray code. In their terms, our codes would be {\em cyclic} (quasi)-Gray codes.
If $\ell =|\D|$, we call the codes \emph{space-optimal} (quasi-)Gray codes.

\paragraph*{Decision Assignment Tree:} The computational model we
consider in this paper is called \emph{Decision Assignment Tree} (DAT).
The definition we provide below is a generalization of that given
in~\cite{Fred78}. It is intended to capture random access machines
with small word size.

Let us fix an underlying domain $\D^n$ whose
elements we wish to enumerate. In the following,
we will denote an element in $\D^n$ by
$\< x_1,x_2,\ldots ,x_n\>$. 
A decision assignment tree is a $|\D|$-ary tree such that each internal
node is labeled by one of the variables $x_1,x_2,\ldots ,x_n$.  
Furthermore, each outgoing edge of an internal node is
labeled with a distinct element of $\D$.
Each leaf node of the tree is labeled by a set of
assignment instructions that set new (fixed) values to chosen variables.
The variables which are not mentioned in the assignment instructions
remain unchanged. 

The execution on a decision assignment tree on a particular input vector $\< x_1,x_2,\ldots ,x_n\> \in \D^n$ starts 
from the root of the tree and
continues in the following way: at a non-leaf node labeled with
a variable $x_i$, the execution queries $x_i$ and depending on the value of $x_i$
the control passes to the node following the outgoing edge labeled
with the value of $x_i$.
Upon reaching a leaf, the corresponding set of assignment statements
is used to modify the vector $\< x_1,x_2,\ldots ,x_n\>$ and the execution terminates.
The modified vector is the output of the execution.

Thus, each decision assignment tree computes a mapping from $\D^n$ into $\D^n$. We are interested in decision assignment trees computing the mapping 
$next(C,\< x_1,x_2,\ldots ,x_n\>)$ for some counter $C$. When $C$ is space-optimal
we can assume, without loss of generality, that each leaf assigns values only to the variables that
it reads on the path from the root to the leaf. (Otherwise, the decision assignment tree does not compute a bijection.)
We define the {\em read complexity} of a decision assignment tree $T$, denoted by $\Read(T)$,
as the maximum number of non-leaf nodes along any path from the root
to a leaf.
Observe that any mapping from $\D^n$ into $\D^n$
can be implemented by a decision assignment tree with read complexity $n$.
We also define the {\em write complexity} of a decision assignment tree $T$,
denoted by $\Write(T)$, as the maximum number of assignment
instructions in any leaf.

Instead of the domain $\D^n$, we will sometimes also use domains that are a cartesian product of different domains. The definition of a decision assignment tree
naturally extends to this case of different variables having different domains.

For any counter $C=(w_1,\ldots,w_\ell)$, we say that $C$ is computed by
a decision assignment tree $T$ if and only if for $i \in [\ell ]$, $next(C,w_i)=T(w_i)$,
where $T(w_i)$ denotes the output string obtained after
an execution of $T$ on $w_i$.  
Note that any two consecutive strings in the cyclic sequence of
$C$ differ by at most $\Write(T)$ many coordinates.

For a small constant $c\ge 1$, some domain $\D$, and all large enough $n$, we will be interested in constructing
cyclic counters on $\D^n$ that are computed by decision assignment trees of write complexity $c$ and read complexity $O(\log n)$.
By the definition such cyclic counters will necessarily be $c$-Gray codes.

\subsection{Construction of Gray codes}
\label{sec:const-gray-code}
For our construction of quasi-Gray codes on a domain $\D^n$ with decision assignment trees of small read and write complexity we will need 
ordinary Gray codes on a domain $\D^{O(\log n)}$.
Several constructions of space-optimal binary Gray codes are known where the oldest one is the binary reflected Gray code~\cite{Gray53}. 
This can be generalized to space-optimal (cyclic) Gray codes over non-binary alphabets (see e.g. \cite{Cohn63, Knuth11}). 
\begin{theorem}[\cite{Cohn63, Knuth11}]
\label{thm:general-Graycode}
For any $m,n \in \N$, there is a space-optimal (cyclic) Gray code over $\Z_m^n$.
\end{theorem}

\section{Chinese Remainder Theorem for Counters}
\label{sec:chinese}
In this paper we consider quasi-Gray codes over
$\Z_m^n$ for $m \in \N$.
Below we describe how to compose decision assignment trees over different domains
to get a decision assignment tree for a larger mixed domain.

\begin{theorem}[Chinese Remainder Theorem for Counters]
\label{thm:chinese}
Let $r,n_1,\dots,n_r \in \N$ be integers, and let $\D_{1,1},\dots,\D_{1,n_1},\D_{2,1},\dots,\D_{r,n_r}$
be some finite sets of size at least two.
Let $\ell _1 \ge r-1$ be an integer, and $\ell _2,\ldots,\ell _r$ be pairwise co-prime integers.
For $i \in [r]$, let $C_i$ be a counter of length $\ell_i$ over
$\D_i=\D_{i,1}\times\cdots\times\D_{i,n_i}$
computed by a decision assignment tree $T_i$ over $n_i$ variables.
Then, there exists a decision assignment tree $T$  over
$\sum_{i=1}^r n_i$ variables that
implements a counter $C$ of length $\prod_{i=1}^r \ell _i$
over $\D_1 \times \cdots\times \D_r$.
Furthermore,  $\Read(T)=n_1 + \max\{\Read(T_i)\}_{i=2}^{r}$,
and $\Write(T)=\Write(T_1) + \max\{\Write(T_i)\}_{i=2}^{r}$. 
\end{theorem}
\begin{proof}
For any $i \in [r]$, let the counter $C_i=(w_{i,1},\ldots,w_{i,\ell _i})$.
Let $x_1,\ldots,x_r$ be variables taking values
in $\D_1,\ldots,\D_r$, respectively.
The following procedure, applied repeatedly, defines the counter $C$:
\begin{align*}
\text{\bf If }\;\;\;\; & x_1=w_{1,i} \text{ for some } i\in[r-1] \text{ \bf then}\\
& x_{i+1} \gets next(C_{i+1},x_{i+1}) \\
& x_1 \gets next(C_1,x_1) \\
\text{\bf else } & \\
& x_1 \gets next(C_1,x_1).  
\end{align*}

It is easily seen that the above procedure defines a valid
cyclic sequence when starting at $w_{1,i_1}, \ldots ,w_{r,i_r}$
for any $\<i_1,i_2,\ldots ,i_r\>\in [\ell_1]\times \cdots \times [\ell_r]$.
That is, 
every element has a unique predecessor and a unique successor,
and that the sequence is cyclic. 
It can easily be implemented
by a decision assignment tree, say $T$. First it reads the value of $x_1$. Since
$x_1 \in \D_1 = \D_{1,1}\times\cdots\times\D_{1,n_1}$,
it queries $n_1$ components. Then, depending on the value of $x_1$,
it reads and updates another component, say $x_j$. This can be
accomplished using the decision assignment tree $T_j$. We also update the value of
$x_1$, and to that end we use the appropriate assignments from decision assignment tree $T_1$. Observe that
irrespective of how efficient $T_1$ is, we read $x_1$
completely to determine which of the remaining $r-1$ counters to update. 
Hence, $\Read(T)=n_1 +\max\left\{\Read(T_i)\right\}_{i=2}^{r}$,
and $\Write(T)=\Write(T_1) + \max\left\{\Write(T_i)\right\}_{i=2}^{r}$.

Now it only remains to show that the counter described above
is indeed of length $\prod_{i=1}^r \ell _i$. Thus, it suffices
to establish that starting with the string
$\<w_{1,1},\ldots,w_{r,1}\>$, 
we can generate the string $\<w_{1,i_1},\ldots,w_{r,i_r}\>$
for any $\<i_1,\ldots,i_r\> \in [\ell_1]\times \cdots \times [\ell_r]$. 
Let us assume $i_1=1$. 
At the end of the proof we will remove this assumption.
Suppose the string $\<w_{1,1},w_{2,i_2},\ldots,w_{r,i_r}\>$
is reachable from $\<w_{1,1},w_{2,1},\ldots,w_{r,1}\>$
in $t$ steps.
As our procedure always increment $x_1$, $t$ must be
divisible by $\ell_1$. Let $d=t /\ell_1$.
Furthermore, the procedure increments a variable $x_i$, $i\ne 1$,
exactly after $\ell_1$ steps. Thus,
$\<w_{1,1},w_{2,i_2},\ldots,w_{r,i_r}\>$ is reachable if and only if 
$d$ satisfies the following equations:
\begin{align*}
d & \equiv i_2-1 \pmod{\ell_2} \\
d & \equiv i_3-1 \pmod{\ell_3} \\
& ~~~~~~~\vdots \\
d & \equiv i_r-1 \pmod{\ell_r}.
\end{align*}
Since $\ell_2,\ldots,\ell_r$ are pairwise co-prime,
Chinese Remainder Theorem (for a reference, see~\cite{DPS96})
guarantees the existence of a unique integral solution  $d$
such that $0 \le d < \prod_{i=2}^r \ell_i$. Hence,
$\<w_{1,1},w_{2,,i_2},\ldots,w_{r,i_r}\>$ is reachable from
$\<w_{1,1},w_{2,1},\ldots,w_{r,1}\>$ in at most $\prod_{i=1}^r\ell_i$
steps. 

Now we remove the assumption $i_1 =1$, i.e., $w_{1,i_1} \neq w_{1,1}$.
Consider the string $\<w_{1,1},w_{2,i_{2}^{'}},\ldots,$ $w_{r,i_{r}^{'}}\>$ 
where $w_{j,i_{j}^{'}} = w_{j,i_{j}-1}$
for $2 \leq j \leq \min\{i_1, r\}$, and $w_{j,i_{j}^{'}} = w_{j,i_{j}}$
for $j > \min\{i_1, r\}$. From the arguments in the previous
paragraph, we know that this tuple is reachable. We now observe that
the next $i_1 -1$ steps increment $w_{1,1}$ to $w_{1,i_1}$ and
$w_{j,i_{j}^{'}}$ to $w_{j,i_j}$ for $2 \leq j \leq \min\{i_1,r\}$,
thus, reaching the desired string $\<w_{1,i_1},\ldots,w_{r,i_r}\>$.
\end{proof}

\begin{remark}
\label{rem:CRT-counters}
We remark that if $C_i$'s are space-optimal in
Theorem~\ref{thm:chinese}, then so is $C$. 
\end{remark}

In the above proof, we constructed a special type of a counter
where we always read the first coordinate, incremented it,
and further depending on its value, we may  
update the value of another coordinate.
From now on we refer to such type of counters as
\emph{hierarchical counters}. In Section~\ref{sec:LB-hierarchy}
we will show that for such type of a counter the co-primality
condition is necessary at least for $\ell _1= 2,3$.
One can further note that the above theorem is similar to
the well known Chinese Remainder Theorem and has similar type of
application for constructing of space-optimal
quasi-Gray codes over $\Z_m^n$
for arbitrary $m \in \N$.
	
\begin{lemma}
\label{lem:stitch-counter}
Let $n,m \in \N$ be such that $m = 2^k \om$, where
$\om$ is odd and $k \geq 0$. Given decision assignment trees $T_1$ and $T_2$ 
computing space-optimal (quasi-)Gray codes over
$\left(\Z_{2^{k}}\right)^{n-1}$ and
$\Z_{\om}^{n-1}$, respectively, there exists a decision assignment tree $T$
implementing a space-optimal quasi-Gray code over
$\Z_m^n$ such that
$\Read(T) = 1 + \max \{\Read(T_1),\Read(T_2)\}$, 
and $\Write(T) = 1 + \max \{\Write(T_1),\Write(T_2)\}$.

\end{lemma}
\begin{proof}
We will view $\Z_m^n$ as $Z_m \times \left(\Z_{2^{k}}\right)^{n-1} \times \left(\Z_{o}\right)^{n-1}$ 
and simulate a decision assignment tree operating on $Z_m \times \left(\Z_{2^{k}}\right)^{n-1} \times \left(\Z_{o}\right)^{n-1}$ 
on $\Z_m^n$.
From the Chinese Remainder Theorem (see~\cite{DPS96}), we know
that there exists a bijection (in fact, an isomorphism)
$f\colon \Z_m \to\Z_{2^k}\times\Z_{\om}$.
We denote the tuple $f(z)$ by $\<f_1(z),f_2(z)\>$. 
From Theorem~\ref{thm:chinese} we know that there exists a
decision assignment tree $T'$ over $\Z_m \times \left(\Z_{2^k}\right)^{n-1} \times \left(\Z_{\om}\right)^{n-1}$ computing a space-optimal
quasi-Gray code such that
$\Read(T') = 1 + \max \{\Read(T_1),\Read(T_2)\}$, 
and $\Write(T') = 1 + \max \{\Write(T_1),\Write(T_2)\}$.

We can simulate actions of $T'$ on an input $\Z_m^n$ to obtain the desired decision assignment tree $T$. 
Indeed, whenever $T'$ queries $x_1$, $T$ queries the first coordinate of its input. Whenever $T'$ queries the $i$-th coordinate 
of $\left(\Z_{2^k}\right)^{n-1}$, $T$ queries the $(i+1)$-th coordinate of its input and makes its decision based on the $f_1(\cdot)$
value of that coordinate. Similarly, whenever $T'$ queries the $j$-th coordinate 
of $\left(\Z_{\om}\right)^{n-1}$, $T$ queries the $(j+1)$-th coordinate and makes its decision based on the $f_2(\cdot)$
value of that coordinate. Assignments by $T$ are handled in similar fashion by updating only the appropriate part of $\<f_1(z),f_2(z)\>$.
(Notice, queries made by $T$ might reveal more information than queries made by $T'$.)
\end{proof}

Before proceeding further, we would also like to point out that
to get a space-optimal decision assignment tree over $\Z_{2^{k}}$,
it suffices to a get space-optimal decision assignment trees over $\Z_2$
for arbitrary dimensions. Thus, to get a decision assignment tree implementing
space-optimal quasi-Gray codes over $\Z_m$, we only
need decision assignment trees implementing space-optimal quasi-Gray codes over
$\Z_2$ and $\Z_{\om}$. This also justifies
our sole focus on construction of space-optimal decision assignment trees over
$\Z_2$ and $\Z_{\om}$ in the later sections. 

\begin{lemma}
\label{clm:stitch-binary}
If, for all $n \in \N$, there exists a decision assignment tree $T$ implementing
a space-optimal (quasi-)Gray code over $\Z_2^n$,
then for any $k$ and $n \in \N$, there exists a decision assignment tree $T'$
implementing a space-optimal (quasi-)Gray code over
$\left(\Z_{2^k}\right)^n$ such that
the read and write complexity remain the same.
\end{lemma}
\begin{proof}
Consider any bijective map
$f : \Z_{2^k} \to \Z_2^k$.
For example, one can take standard binary encoding of integers
ranging from $0$ to $2^k -1$ as the bijective map $f$.
Next, define another map
$g : \left(\Z_{2^k }\right)^n \to \Z_2^{kn}$ as follows:
$g(x_1,\ldots ,x_n) = \<f(x_1),\ldots ,f(x_n)\>$.
Now consider $T$ that implements a space-optimal (quasi-)Gray code
over $\Z_2^{k n}$. We fix a partition
of the variables
$\{1,\ldots ,k\} \uplus \cdots \uplus \{(n-1)k +1 , \ldots ,nk \}$
into $n$ blocks of $k$ variables each. 

We now construct a decision assignment tree $T'$ over $\left(\Z_{2^k}\right)^n$
using $T$ and the map $f$. As in the proof of
Lemma~\ref{lem:stitch-counter}, our $T'$ follows $T$
in the decision making. That is, if $T$ queries a variable,
then $T'$ queries the block in the partition where the variable lies.
(Again, as noted before, $T'$ may get more information than required
by $T$.) Upon reaching a leaf, using $f$,
$T'$ updates the blocks depending on $T$'s updates to the variables. 

\end{proof}

We devote the rest of the paper to the construction of counters
over $\Z_2^n$, and $\Z_m^n$
for any odd $m$.

\section{Permutation Group and Decomposition of Counters}
\label{sec:permutation-counter}
	We start this section with some basic notation and facts about
the permutation group which we will use heavily in the rest of the paper.
The set of all permutations over a domain
$\D$ forms a group under the composition operation, 
denoted by $\circ$, which is defined as follows:
for any two permutations $\sigma$ and $\alpha$,
$\sigma \circ \alpha(x) = \sigma(\alpha(x))$, where
$x \in \D$. The corresponding group,
denoted $\mathcal{S}_{N}$, is the \emph{symmetric group}
of order $N=|\D|$. We say, a permutation
$\sigma \in \mathcal{S}_N$ is a \emph{cycle} of length $\ell$ if there are distinct elements $a_1,\dots,a_\ell \in [N]$ 
such that for $i \in [\ell -1]$, $a_{i+1}=\sigma(a_i)$, $a_1=\sigma(a_\ell)$, and for all
$a\in [N] \setminus \{a_1, a_2, \ldots , a_\ell\}$, 
$\sigma(a)=a$.  We denote such a cycle by $(a_1,a_2,\cdots ,a_\ell)$.
Below we state few simple facts about composition of cycles.

\begin{proposition}
\label{prop:intersect-cycle}
Consider two cycles $C_1=(t, a_1, \cdots, a_{\ell_1})$ and
$C_2=(t, b_1, \cdots, b_{\ell_2})$ where for any $i \in [\ell_1]$
and $j\in [\ell_2]$, $a_i \ne b_j$.
Then, $C = C_2 \circ C_1$ is the cycle $(t, a_1, \cdots, a_{\ell_1}, b_1, \cdots, b_{\ell_2})$ of length $\ell_1 + \ell_2 +1$.
\end{proposition}

\begin{proposition}
\label{prop:cycle-conjugate}
If $\sigma \in \mathcal{S}_N$ is a cycle of length $\ell$,
then for any $\alpha \in \mathcal{S}_N$,
$\alpha\circ \sigma\circ \alpha^{-1}$ is also a cycle of length
$\ell$. Moreover, if $\sigma=(a_1, a_2, \cdots, a_\ell)$,
then $\alpha\circ \sigma\circ \alpha^{-1}= (\alpha(a_1), \alpha(a_2), \cdots, \alpha(a_\ell))$.
\end{proposition}

The permutation $\alpha\circ \sigma\circ \alpha^{-1}$ is called the
\emph{conjugate} of $\sigma$ with respect to $\alpha$.
The above proposition is a special case of a well known fact
about the cycle structure of conjugates of any permutation
and can be found in any standard text book on Group Theory
(e.g., Proposition~$10$ in Chapter~$4.3$ of~\cite{DF04}.).

Roughly speaking, a counter of length $\ell$ over $\D$,
in the language of permutations, is nothing but 
a cycle of the same length in $\mathcal{S}_{|\D|}$. 
We now make this correspondence precise and give a construction of
a decision assignment tree that implements such a counter.

\begin{lemma}
\label{lem:cycle-counter}
Let $\D=\D_1 \times \cdots \times \D_{r}$
be a domain. Suppose
$\sigma_1,\ldots,\sigma_k \in \mathcal{S}_{|\D|}$
are such that
$\sigma=\sigma_k \circ \sigma_{k-1}\circ \cdots \circ \sigma_1$
is a cycle of length $\ell$. Let  $T_1,\ldots ,T_k$
be decision assignment trees that implement $\sigma_1,\ldots,\sigma_k$ respectively. 
Let
$\mathcal{D'}=\mathcal{D'}_1 \times \cdots \times \mathcal{D'}_{r'}$ be a domain
such that $|\mathcal{D'}| \geq k$, and let $T'$ be a decision assignment tree
that implements a counter $C'$ of length $k'$ over $\mathcal{D'}$
where $k' \geq k$. 

Then, there exists a decision assignment tree $T$ that implements a counter of length
$k' \ell$ over $\mathcal{D'}\times \D$ such that
$\Read(T) = r' + \max \{\Read(T_i)\}_{i=1}^{k}$, 
and $\Write(T) = \Write(T') + \max \{\Write(T_i)\}_{i=1}^{k}$.
\end{lemma}
\begin{proof}
Suppose $C'=(a_1,\ldots,a_{k'})$. Now let us consider the following procedure $P$:
on any input $\<x_1,x_2\> \in \mathcal{D'} \times \D$,
\begin{align*}
\text{\bf If }\;\;\;\; & x_1=a_j \text{ for some } j\in[k] \text{ \bf then}\\
& x_{2} \gets \sigma_j(x_2) \\
& x_1 \gets next(C',x_1) \\
\text{\bf else } & \\
& x_1 \gets next(C',x_1).  
\end{align*}

Now using a similar argument as in the proof of
Theorem~\ref{thm:chinese}, the above procedure is easily seen
to be implementable using a decision assignment tree $T$ of the prescribed complexity.
Each time we check the value of
$x_1 \in \mathcal{D'}=\mathcal{D'}_1 \times \cdots \times \mathcal{D'}_{r'}$.
Thus, we need to read $r'$ components.
Depending on the value of $x_1$,
we may apply $\sigma_j$ on $x_2$ using the decision assignment tree $T_j$.
Then we update the value of $x_1$.
Hence, $\Read(T)=r'+\max\{\Read(T_i)\}_{i=1}^{k}$,
and $\Write(T)=\Write(T') + \max\{\Write(T_i)\}_{i=1}^{k}$.

Let $(w_1,w_2\cdots, w_\ell)$
be the cycle of length
$\ell$ given by $\sigma$.
We now argue that the procedure $P$ generates a counter of length
$k' \ell$ over $\mathcal{D'} \times \D$ starting at
$\<a_1,w_1\>$. 
Without loss of generality, let us assume that
$\sigma=\sigma_{k'}\circ \cdots \circ \sigma_{k+1}\circ \sigma_k \circ \sigma_{k-1}\circ \cdots \circ \sigma_1$, 
	where for $j \geq k+1$, $\sigma_j$ is the identity map. 
Fix $j \in [k']$. 
Define $\alpha_j=\sigma_{j-1} \circ \cdots \circ \sigma_{1}$, and 
$\tau_j = \alpha_j \circ \sigma \circ \alpha_j^{-1} = \sigma_{j-1} \circ \cdots \circ \sigma_1\circ \sigma_{k'}\circ \cdots \circ  \sigma_{j}$.
For $i=0,1,\dots, \ell$, let $\<g_i,e_i\>=P^{ik'}(\<a_j,\alpha_j(w_1)\>)$ 
where $P^{ik'}$ denotes $ik'$ invocations of $P$. Since $P$ increments
$x_1$ in every invocation, for $i=1,2,\dots,\ell$, $g_i=a_j$ and $e_{i}=\tau_j(e_{i-1})$.

By Proposition~\ref{prop:cycle-conjugate},
$\tau_j$ is a cycle $(\alpha_j(w_1)\alpha_j(w_2)\cdots \alpha_j(w_\ell))$ of length $\ell$.
Hence, $e_1,\ldots , e_\ell$ are all distinct and $e_\ell=e_0$.

As a consequence we conclude that for any
$x \in \mathcal{D'} \times \D$ and
$1 \leq j_1 \neq j_2 \leq k' \ell$, $P^{j_1}(x) \ne P^{j_2}(x)$
and $P^{k' \ell}(x)=x$. This completes the proof.
\end{proof}

In the next two sections we describe the construction of $\sigma_1,\cdots,\sigma_k \in \mathcal{S}_N$ where $N=m^n$ for some $m,n \in \N$ and how the value of $k$ depends on the length of the cycle $\sigma=\sigma_k \circ \sigma_{k-1}\circ \cdots \circ \sigma_1$.

\section{Counters via Linear Transformation}
\label{sec:binarycounter}

The construction in this section is based on linear transformations.
Consider the vector space $\mathbb{F}_q^n$, and
let $L : \mathbb{F}_q^n \to \mathbb{F}_q^n$ be a linear transformation. A basic fact in linear algebra says that if $L$ has \emph{full} rank,
then the mapping given by $L$ is a bijection. Thus,
when $L$ is full rank, the mapping can also be thought of as a
permutation over $\mathbb{F}_q^n$. Throughout this section we use many basic terms related to linear transformation without defining them, for the details of which we refer the reader to any standard text book on linear algebra (e.g.~\cite{Lang87}).

A natural way to build counter out of a full rank linear transformation
is to fix a starting element, and repeatedly apply
the linear transformation to obtain the next element.
Clearly this only list out elements in the cycle
containing the starting element. Therefore, we would like to choose  
the starting element such that we enumerate the largest cycle. Ideally,
we would like the largest cycle to contain all the elements of
$\mathbb{F}_q^n$. However this is not possible because
any linear transformation fixes the all-zero vector.
But there do exist full rank linear transformations such that
the permutation given by them is a single cycle of length $q^n-1$.
Such a linear transformation would give us a counter over a domain
of size $q^n$ that enumerates all but one element. Clearly, a trivial
implementation of the aforementioned argument would lead to a counter
that reads and writes all $n$ coordinates in the worst-case. In the
rest of this section, we will develop an implementation and argue about
the choice of linear transformation such that the read and write
complexity decreases exponentially. 

We start with recalling some basic facts from linear algebra.

\begin{definition}[Linear Transformation]
\label{def:lt}
A map $L : \mathbb{F}_q^n \to \mathbb{F}_q^n$ is called
a linear transformation if 
$L(c\cdot x + y) = cL(x) + L(y)$,
for all $x, y \in \mathbb{F}_q^n$ and $c \in \mathbb{F}_q$.
\end{definition}
It is well known that every linear transformation $L$ is associated with some matrix $A \in \mathbb{F}_q^{n \times n}$ such 
that applying the linear transformation is equivalent to the left multiplication by $A$.
That is, $L(x) = Ax$ where we interpret $x$ as a column vector. Furthermore, $L$ has full rank iff $A$ is invertible over $\mathbb{F}_q$. 

\begin{definition}[Elementary matrices]
\label{def:em}
An $n \times n$ matrix over a field $\mathbb{F}$ is said to be an
\emph{elementary matrix} if it has one of the following forms:
\begin{itemize}
\item[(a)] The off-diagonal entries are all $0$. For some $i \in [n]$,
($i,i$)-th entry is a non-zero $c \in \mathbb{F}$.
Rest of the diagonal entries are $1$. For a fixed $i$, we denote
all matrices of this type by $E_{i,i}$. (See Fig.~\ref{fig:em}.)
\item[(b)] The diagonal entries are all 1. For some $i$ and $j$,
$1 \leq i \neq j \leq n$, ($i,j$)-th entry is a non-zero
$c \in \mathbb{F}$. Rest of the off-diagonal entries are $0$.
For each $i$ and $j$, $i\neq j$, we denote all matrices of
this type by $E_{i,j}$. (See Fig.~\ref{fig:em}.)
\end{itemize}
\end{definition}
\begin{figure}[h]
\centering
\[ (a) ~~~~
\bordermatrix{
&    &        &   & i &   &        &   \cr
&  1 &        &   &   &   &        &   \cr
&    & \ddots &   &   &   &        &   \cr
&    &        & 1 &   &   &        &   \cr
i   &    &        &   & c &   &        &   \cr
&    &        &   &   & 1 &        &   \cr 
&    &        &   &   &   & \ddots &   \cr
&    &        &   &   &   &        & 1 
} \qquad \text{ or, } \qquad
(b) ~~~~
\bordermatrix{
&    &        & j  &  &   &        &   \cr
&  1 &        &   &   &   &        &   \cr
&    & \ddots &   &   &   &        &   \cr
&    &        & 1 &   &   &        &   \cr
&    &        &   & \ddots &   &        &   \cr
i    &    &        & c  &   & 1 &        &   \cr 
&    &        &   &   &   & \ddots &   \cr
&    &        &   &   &   &        & 1 
}
\]
\caption{Elementary matrices}
\label{fig:em}
\end{figure}

From the definition it is easy to see that left multiplication by an elementary matrix
of type $E_{i,i}$ is equivalent to multiplying the $i$-th row with $c$,
and by an elementary matrix of type $E_{i,j}$ is equivalent to adding
$c$ times $j$-th row to the $i$-th row. 

\begin{proposition}
\label{prop:em-product}
Let $A \in \mathbb{F}^{n \times n}$ be invertible.
Then $A$ can be written as a product of $k$ elementary matrices such
that $k \leq n^2 + 4(n-1)$. 
\end{proposition}
\begin{proof}
Consider the inverse matrix $A^{-1}$ which is also full rank.
It is easily seen from Gauss elimination that by left multiplying $A^{-1}$ with at most
$n^2$ many elementary
matrices, we can reduce $A^{-1}$ to a permutation matrix. A
permutation matrix is a $\{0,1\}$-matrix that has exactly
one $1$ in each row and column.  Now we need at most $(n-1)$
row swaps to further reduce the obtained permutation matrix
to the identity matrix.  We claim that a row swap can be
implemented by left multiplication with at most $4$
elementary matrices. Indeed, to swap row $i$ and row $j$,
the following sequence of operation suffices: $(i)$ add
$j$-th row to $i$-th row, $(ii)$ subtract $i$-th row from
$j$-th row, $(iii)$ add $j$-th row to $i$-th row, and $(iv)$
multiply $j$-th row with $-1$. (The last operation is not required 
if the characteristic of the underlying field is $2$.)

Hence, the inverse of $A^{-1}$ which is our original matrix $A$ is
the product of $k$ elementary matrices.
\end{proof}

\subsection{Construction of the counter}
\label{subsec:generic-construction}

Let $A$ be a full rank linear transformation from  $\mathbb{F}_q^n$ to
$\mathbb{F}_q^n$ such that when viewed as permutation it is a single
cycle of length $q^n-1$. More specifically, $A$ is an invertible matrix in $\mathbb{F}_q^{n \times n}$ such that for any $x\in \mathbb{F}_q^n$ where $x\ne(0,\ldots,0)$, $Ax,A^2x,\ldots,A^{(q^n-1)}x$ are distinct. Such a matrix exists, for example, take $A$ to be the matrix of a linear transformation that corresponds to multiplication from left by a fixed generator of the multiplicative group of $\mathbb{F}_{q^n}$ under the standard vector representation of elements of $\mathbb{F}_{q^n}$. Let $A = E_kE_{k-1}\cdots E_1$ where $E_i$'s are elementary matrices.

\begin{theorem}
\label{thm:lt-counter}
Let $q$, $A$, and $k$ be as defined above. Let $r \ge  \log_q k$. There exists a quasi-Gray code
on the domain $(\mathbb{F}_q)^{n+r}$ of length $q^{n+r} - q^r$
that can be implemented using a decision assignment tree $T$ such that
$\Read(T) \leq r+2$ and $\Write(T) \leq 2$.  
\end{theorem}
\begin{proof}
The proof follows readily from Lemma~\ref{lem:cycle-counter},
where $E_i$'s play the role of $\sigma_i$'s, and noting that the permutation given by
any elementary matrix can be implemented using a decision assignment tree that reads
at most two coordinates and writes at most one. For the counter $C'$ on $(\mathbb{F}_q)^{r}$
we chose a Gray code of trivial read complexity $r$ and write complexity 1.
\end{proof}

Thus, we obtain a counter on a domain of size roughly $k q^n$ that
misses at most $qk$ elements. Clearly, we would like to minimize $k$.
A trivial bound on $k$ is $O(n^2)$ that follows from Proposition~\ref{prop:em-product}. We now discuss the choice of
$A$ so that $k$ becomes $O(n)$. We start with recalling a notion
of primitive polynomials over finite fields. 

\begin{definition}[Primitive polynomials]
\label{def:primitive}
A \emph{primitive} polynomial $p(z) \in \mathbb{F}_q[z]$
of degree $n$ is
a monic irreducible polynomial over $\mathbb{F}_q$ such that any
root of it in $\mathbb{F}_{q^n}$ generates the multiplicative
group of $\mathbb{F}_{q^n}$. 
\end{definition}

\begin{theorem}[\cite{LN96}]
\label{thm:existence-primitive}
The number of primitive polynomials of degree $n$ over $\mathbb{F}_q$
equals $\phi(q^n - 1)/n$,
where $\phi(\cdot)$ is the Euler $\phi$-function.
\end{theorem}

Let $p(z)$ be a primitive polynomial of degree $n$ over $\mathbb{F}_q$.
The elements of $\mathbb{F}_{q^n}$ can be uniquely expressed
as a polynomial in $z$ over $\mathbb{F}_q$ of degree
at most $n-1$. 
In particular, we can identify an element of $\mathbb{F}_{q^n}$ with
a vector in $\mathbb{F}_q^n$ that is given by the 
coefficient vector of the unique polynomial expression
of degree at most $n-1$.
But, since $p(z)$ is primitive, we also know that
$\mathbb{F}_{q^n} = \{0,1,z,z^2,\ldots , z^{q^n-2}\}$.
This suggests a particularly nice linear transformation to consider:
\emph{left multiplication by $z$}. This is so because
the matrix $A$ of the linear transformation computing the multiplication by $z$ is very \emph{sparse}.
In particular,
if $p(z) = z^n + c_{n-1}z^{n-1}+c_{n-2}z^{n-2}+ \cdots + c_1z+ c_0$,
then $A$ looks as follows:
\[
\begin{pmatrix}
-c_{n-1} & 1      & 0      & \cdots & 0 & 0 & 0          \\
-c_{n-2} & 0      &  1     & \cdots & 0 & 0 & 0      \\
-c_{n-3} & 0      &  0     & \cdots & 0 & 0 & 0       \\
\vdots  & \vdots & \vdots &  \ddots & \vdots &\vdots &\vdots \\
-c_2    &  0     &  0     & \cdots & 0 & 1 & 0          \\
-c_1    &  0     & 0      & \cdots & 0 & 0 & 1         \\
-c_0    &  0     & 0      & \cdots & 0 & 0 & 0
\end{pmatrix}
.\]
Thus, from the proof of Proposition~\ref{prop:em-product}, it follows
that $A$ can be written as a product of at most $n + 4(n-1)$
elementary matrices. (When $q$ is a power of $2$, then the number
of elementary matrices in the product is at most $n+3(n-1)$.)
Hence, from the discussion above and
using Theorem~\ref{thm:lt-counter}, we obtain
the following corollaries. Setting $r=\lceil \log (4n -3) \rceil$ in Theorem~\ref{thm:lt-counter}
gives:

\begin{corollary}
\label{cor:linear-counter-f2}
For any $n' \geq 2$, and $n=n'+\lceil \log (4n' -3) \rceil$, there exists a counter on
$(\Z_2)^{n}$ that misses at most $8n$ strings and can be implemented by a decision assignment tree that
reads at most $4+\log n$ bits and writes at most $2$ bits.
\end{corollary}

By doubling the number of missed strings and increasing the number of read bits by one we can
construct given counters for any $\Z_2^{n}$, where $n\ge 15$.
In the above corollary the number of missed strings grows linearly with $n$.
One might wonder if it is possible to keep the number of missing strings
$o(n)$, while keeping the read complexity essentially the same.
The next corollary shows that this is indeed possible,
but at the cost of increasing the write complexity.

\begin{corollary}
\label{cor:sublinear-counter-f2}
For $n \geq 2$, there exists a counter on $(\Z_2)^{n+O(\log n)}$
that misses out at most $O(n/\lceil \log n \rceil)$ strings.
Furthermore, it can be implemented by a decision assignment tree that reads and writes
at most $O(\log n)$ bits.
\end{corollary}
\begin{proof}
The idea is simply to increase the underlying alphabet size.
In particular, let $q = 2^{\lceil \log n \rceil}$ in
Theorem~\ref{thm:lt-counter}.  
\end{proof}

We also remark that by taking $q$ to be $2^{\frac{n}{C}}$,
where $C > 1$ is a universal constant, one would get a counter on
$(\Z_2)^{n+O(1)}$ that misses only $O(C)$ strings
(i.e. independent of $n$). However, the read and write complexity
gets worse. They are at most $2\left(\frac{n}{C}\right) + O(1)$.

For the general case, when $q$ is a prime power, we obtain
the following corollary by setting $r$ to $\lceil \log_q (5n-4) \rceil$ or $1+\lceil \log_q  (5n-4) \rceil$ in Theorem~\ref{thm:lt-counter}.

\begin{corollary}[Generalization of Theorem~\ref{thm:binary-counter}]
\label{cor:counter-prime-power}
Let $q$ be any prime power. For $n \geq 15$, there exists a counter
on $\Z_q^{n}$ that misses at most $5q^2n$ strings and that is computed by 
a decision assignment tree with read complexity at most $6 + \log_q n$ and write complexity $2$. 
\end{corollary}

\begin{remark}
\label{rem:finding-primitive-polynomials}
We remark that the algorithm given in this section can be made
uniform. To achieve that we need to obtain a primitive
polynomial of degree $n$ over $\mathbb{F}_q$ uniformly.
To this end, we can use a number of algorithms
(deterministic or probabilistic);
for example,~\cite{Shoup92,Shparlinski93,Shparlinski96}.
For a thorough treatment,
we refer to Chapter~2 in~\cite{Shparlinski99}.
\end{remark}

\section{Space-optimal Counters over $\Z_m^n$ for any Odd $m$}
\label{sec:oddcounter}
This section deals with space-optimal counter over an odd-size alphabet. 
We start by recalling Theorem~\ref{thm:odd-counter} in terms of decision assignment tree complexity.
\begin{theorem}[Restatement of Theorem~\ref{thm:odd-counter}]
\label{thm:odd-counter-restated}
For any odd $m \in \N$ and any positive integer $n \ge 15$, there is a space-optimal $2$-Gray code over $\Z_m^n$ that can be computed by a decision assignment tree $T$ such that $\Read(T)\le 4\log_m n$. 
\end{theorem}
Before providing our construction of the quasi-Gray code, we give a short overview of how the construction will proceed. First we set $n'=n-c\cdot \log n$ for some constant $c>0$ that will be fixed later. 
Then we define suitable permutations $\alpha_1,\ldots,\alpha_{n'} \in \mathcal{S}_N$ where $N=m^{n'}$ such that their composition $\alpha_{n'}\circ \cdots \circ \alpha_1$ is a cycle of length $m^{n'}$. Next we show that each $\alpha_i$ can be further decomposed into $\alpha_{i,1},\ldots,\alpha_{i,j} \in \mathcal{S}_N$ for some $j$, such that each $\alpha_{i,r}$ for $r \in [j]$ can be computed using a decision assignment tree with read complexity $3$ and write complexity $1$. Finally to complete the construction we use Lemma~\ref{lem:cycle-counter} with $\alpha_{i,r}$'s playing the role of $\sigma_1,\ldots,\sigma_k$ in the lemma. 

We recall the notion of $r$-functions over $\Z_m^n$ that was introduced by Coppersmith and Grossman~\cite{CG75} for $m=2$. Below we generalize that definition for any $m \in \N$. 
\begin{definition}
\label{def:r-function}
For any $r \in [n-1]$, an \emph{$r$-function} on $\Z_m^n$ is a permutation $\tau$ over $\Z_m^n$ identified by a subset $\{i_1,\ldots,i_r,j\} \subseteq [n]$ of size $r+1$ and a function $f:\Z_m^r \to \Z_m$ such that for any $\<a_1,\ldots,a_n\> \in \Z_m^n$,
$$\tau(\<a_1,\ldots,a_n\>)=\<a_1,\ldots , a_{j-1},a_j+f(a_{i_1},\ldots,a_{i_r}),a_{j+1},\ldots ,a_n\>.$$
\end{definition}
Observe that any $r$-function can be implemented using a decision assignment tree $T$ that queries $x_{i_1},\ldots,x_{i_r}$ and $x_j$ at internal nodes, and at leaves it assigns value only to the variable $x_j$. Thus, $\Read(T)=r+1$ and $\Write(T)=1$.
\begin{claim}
\label{obs:r-function-DAT} 
Any $r$-function on $\Z_m^n$ can be implemented using a decision assignment tree $T$ with $n$ variables such that $\Read(T)=r+1$ and $\Write(T)=1$.
\end{claim}

We are now ready to provide details of our construction of a space-optimal quasi-Gray code over $\Z_m^n$ for any odd $m$. Define $n' :=n-c\cdot \log_m n$ for some constant $c>0$ that will be fixed later.

\subsubsection*{Step 1: Construction of \(\alpha_1,\ldots,\alpha_{n'}\).}

We consider specific permutations $\alpha_1,\ldots ,\alpha_{n'}$ over $\Z_m^{n'}$ such that $\alpha = \alpha_{n'} \circ \cdots \circ \alpha_1$ is a cycle of length $m^{n'}$. We define them below.
\begin{definition}
  \label{def:alpha-i}
  Let $m$ and $n'$ be natural numbers. For $i \in [n']$, we define $\alpha_i$ to be
  the permutation given by the following map: for any $\<x_1,\ldots,x_i,\ldots, x_{n'}\> \in \Z_m^{n'}$,
\[\alpha_i\left(\<x_1,\ldots, x_i,\ldots, x_{n'}\>\right) =
\begin{cases}
  \<x_1,\ldots,x_{i-1}, x_i+1,x_{i+1},\ldots,x_{n'}\> & \text{if } x_j = 0 \text{ for all } j \in [i-1], \\
  \<x_1,\ldots,x_i,\ldots,x_{n'}\> & \text{otherwise.} 
\end{cases}
\]
The addition operation in the mapping $x_i \gets x_i+1$ is within $\Z_m$.
\end{definition}
The following observation is easily seen from the definitions of $r$-functions and $\alpha_i$. 
\begin{claim}
\label{obs:alpha-r-function}
For any $i \in [n']$, $\alpha_i$ is an $(i-1)$-function on $\Z_m^{n'}$.
Furthermore, each $\alpha_i$ is composed of disjoint cycles of length $m$ over $\Z_m^{n'}$.
\end{claim}
We now establish a crucial property of the $\alpha_i$'s, i.e.,
their composition is a \emph{full} length cycle. 
\begin{claim}
\label{clm:full-cycle}
$\alpha = \alpha_{n'} \circ \cdots \circ \alpha_1$ is a cycle of length $m^{n'}$.
\end{claim}
\begin{proof}
  Consider the sequence of permutations $\tau_1,\ldots,\tau_{n'}$ 
  such that $\tau_i=\alpha_i\circ \alpha_{i-1} \circ \cdots \circ \alpha_1$
  for $i \in [n']$. Clearly, $\tau_{n'}=\alpha$. We now prove the claim.
  In fact, we establish the following stronger claim: for $i \in [n']$,
  $\tau_i$ is a permutation composed of $m^{n'-i}$ disjoint cycles, each
  of the cycles being of length $m^i$. Furthermore, for every
  $\<a_{i+1},\ldots,a_{n'}\>\in \Z_m^{n'-i}$ there is a cycle that involves
  all tuples of the form $\<x_1,\ldots,x_{i},a_{i+1},\ldots,a_{n'}\>$.
  The claim, $\alpha$ is a cycle of length $m^{n'}$, follows as a consequence.
  We prove the stronger claim by induction on $i$. 

  \textbf{Base case:} $\tau_1=\alpha_1$. From the definition of $\alpha_1$,
  it follows that there is a cycle of length $m$ of the form
  $\left(\<0,a_2,\ldots,a_{n'}\>,\<1,a_2,\ldots,a_{n'}\>,\ldots,\<m-1,a_2,\ldots,a_{n'}\> \right)$ for each $\<a_2,\ldots,a_{n'}\> \in \Z_m^{n'-1}$.
  Hence our induction hypothesis clearly holds for the base case.

  \textbf{Induction step:} Suppose our induction hypothesis holds until some $i \in [n']$ and
  we would like to establish it for $i+1$. Let us consider the permutation
  $\alpha_{i+1}$. We know that  it is composed of $m^{n'-(i+1)}$ disjoint cycles of
  length $m$. Indeed, for each $\<a_{i+2},\ldots,a_{n'}\> \in \Z_m^{n'-(i+1)}$,
  $\alpha_{i+1}$ contains a cycle that involves all $m$ tuples where
  the first $i$ coordinates are all $0$ and the last $n'-(i+1)$ coordinates
  are set to $\<a_{i+2},\ldots,a_{n'}\>$. From the cycle decomposition of
  $\alpha_{i+1}$ and $\tau_i$ into disjoint cycles,
  it is clear that for any cycle say $C$ in
  $\alpha_{i+1}$ there are $m$ disjoint cycles $C_1,\ldots,C_m$ in $\tau_i$,
  each of them intersecting $C$ in exactly one element.
  Consider $C'=C\circ C_1\circ \cdots \circ C_m$. 
  By repeated application of Proposition~\ref{prop:intersect-cycle},
  we conclude that $C'$ is a cycle of length $\sum_{i=1}^m |C_i|=m^{i+1}$.
  (Here by $|C_i|$ we mean the length of the cycle
  $C_i$.) Also $C'$ involves tuples where the last $n'-(i+1)$ coordinates
  are set to some fixed $\<a_{i+2},\ldots,a_{n'}\> \in \Z_m^{n'-(i+1)}$.
  Thus, $C'$ is a cycle of length $m^{i+1}$ containing all tuples of the form
  $\<x_1,\ldots ,x_{i+1},a_{i+2},\ldots ,a_{n'}\>$. 
  Since $\tau_{i+1}=\alpha_{i+1}\circ \tau_i$ and $\alpha_{i+1}$ contains
  $m^{n'-(i+1)}$ disjoint cycles, we conclude 
  that $\tau_{i+1}$ consists of exactly $m^{n'-(i+1)}$ disjoint cycles,
  each of length $m^{i+1}$ and containing tuples of the required form.

  This finishes the proof.
\end{proof}

Readers may note that this step does not use the fact that $m$ is odd and,
thus, it is true for any $m \in \N$. 
If we were to directly implement $\alpha_i$ by a decision assignment tree,
its read complexity would be $i$.
Hence, we would not get any savings in
Lemma~\ref{lem:cycle-counter}. 
So we need to further decompose  $\alpha_i$ into permutations of small read complexity.

\subsubsection*{Step 2: Further decomposition of $\alpha_i$'s.}

Our goal is to describe $\alpha_i$ as a composition of $2$-functions. Recall,
Claim~\ref{obs:alpha-r-function}, each $\alpha_i$ is an $(i-1)$-function on
$\Z_m^{n'}$. Suppose, for $i \in [n']$, there exists a set of $2$-functions
$\alpha_{i,1},\ldots,\alpha_{i,k_i}$ such that
$\alpha_i=\alpha_{i,k_i}\circ \cdots \circ \alpha_{i,1}$.  Then using
Lemma~\ref{lem:cycle-counter}, where $\alpha_{i,k}$'s play the role of
$\sigma_j$'s, we obtain a decision assignment tree implementing a $2$-Gray code
with potentially low read complexity.
Indeed, each $\alpha_{i,k}$ has low read complexity by Claim~\ref{obs:r-function-DAT}, hence the read complexity
essentially depends on how large is $\sum_i k_i$. 
In the following we will argue that $\alpha_i$'s can be decomposed into a
\emph{small} set of $2$-functions, thus keeping the maximum $k_i$ small.
As a result, the read complexity bound in Theorem~\ref{thm:odd-counter-restated}
will follow.

Note $\alpha_1,\alpha_2$ and $\alpha_3$ are already $2$-functions. In the case of $\alpha_i$, $4 \le i \le n'-2$,
we can directly adopt the technique from~\cite{BC92, BCKLS14} to generate
the desired set of $2$-functions. However, as discussed in Section~\ref{sec:our-technique}, that technique falls short when $i > n'-2$. (It needs two free registers to
operate.) For $i=n'-1$, it is possible to generalize the proof technique of~\cite{CG75} to decompose $\alpha_{n'-1}$. Unfortunately all the previously known techniques fail to decompose $\alpha_{n'}$ and we have to develop a new technique. 

First we provide the adaptation of~\cite{BC92, BCKLS14}, and then develop a new technique that allows us to express both $\alpha_{n'-1}$ and $\alpha_{n'}$ as a composition
of small number of $2$-functions, thus overcoming the challenge described above.

\begin{lemma}
\label{lem:catalytic-decomposition}
For any $4 \le i \le n'-2$, let $\alpha_i$ be the permutation given by Definition~\ref{def:alpha-i}. Then there exists a set of $2$-functions $\alpha_{i,1},\ldots,\alpha_{i,k_i}$ such that $\alpha_i=\alpha_{i,k_i}\circ \cdots \circ \alpha_{i,1}$, and $k_i\le 4 (i-1)^2 - 3$.
\end{lemma}
It is worth noting that, although in this section we consider $m$ to be odd, the above lemma holds for any $m \in \N$. In~\cite{BC92}, computation uses only addition and multiplication from the ring, whereas we can use any function $g : \Z_m \to \Z_m$. This subtle difference makes the lemma to be true for any $m \in \N$ instead of being true only for prime powers.
\begin{proof} Pick $i\le n'-2$.
  Let us represent $\alpha_i$ as an $(i-1)$-function. From the definition we have,
  \[\alpha_i\left(\<a_1,\ldots,a_i,\ldots,a_{n'}\>\right)=\<a_1,\ldots , a_{i-1},a_i+f(a_{1},\ldots,a_{i-1}),a_{i+1},\ldots ,a_{n'}\> ,\]
  where the map $f : \Z_m^{i-1} \to \Z_m$ is defined as follows: 
\[f(a_1,\ldots ,a_{i-1})=
\begin{cases}
1 & \text{if } (a_1,\ldots ,a_{i-1}) =(0,\ldots ,0), \\
0 &\text{otherwise}.
\end{cases}\]

Observe that $f$ is an indicator function of a tuple; in particular, of the all-zeroes tuple. 
To verify the lemma, we would prove a stronger claim than the statement of the lemma. Consider the set $\mathcal{S}$ of $r$-functions, $1 \leq r \leq n'-3$, such that the function $f$ used to define them is the indicator function of the all-zeroes tuple. That is, $\tau \in \mathcal{S}$ if and only if there exists a set $\{i_1,i_2,\ldots ,i_r\} \subseteq [n']$ of size $r$ and a $j \in [n']\setminus \{i_1,i_2,\ldots ,i_r\} $ such that
\[\tau(\<a_1,\ldots,a_{n'}\>)=\<a_1,\ldots , a_{j-1},a_j+f(a_{i_1},\ldots,a_{i_r}),a_{j+1},\ldots ,a_{n'}\>,\]
where $f(x_1,\ldots ,x_r) = 1$ if $(x_1,\ldots ,x_r) = (0,\ldots ,0)$, and $0$ otherwise. Observe that $\alpha_i \in \mathcal{S}$ for $4 \leq i \leq n'-2$.
We establish the following stronger claim.
\begin{claim}
  For an $r$-function $\tau \in \mathcal{S}$, there exist $2$-functions $\tau_1 ,\ldots ,\tau_{k_r}$ such that $\tau = \tau_{k_r}\circ \cdots \circ \tau_1$ and $k_r \leq 4r^2-3$. (We stress that $\tau_1 ,\ldots ,\tau_{k_r}$ need not belong to $\mathcal{S}$.)   
\end{claim}
Clearly, the claim implies the lemma. We now prove the claim by induction on $r$.
The base case is $r \leq 2$, in which case the claim trivially holds.
Suppose the claim holds for all $(r-1)$-functions in $\mathcal{S}$. Let
$\tau \in \mathcal{S}$ be an $r$-function identified by the set
$S := \{i_1,i_2,\ldots ,i_r\} \subseteq [n']$ and $j \in [n'] \setminus S$.  
Since $f$ is an indicator function, it can be expressed as a product of indicator
functions. That is,
$f(a_{i_1},\ldots,a_{i_r})= \prod_{s \in S}g(a_s)$ where
$g : \Z_m \to \Z_m$ is the following $\{0,1\}$-map:   
\[g(y)=
\begin{cases}
1 & \text{if } y=0, \\
0 & \text{otherwise}.
\end{cases}
\]

Consider a partition of $S$ into two sets $A$ and $B$ of sizes $\lfloor r/2 \rfloor$ and $\lceil r/2 \rceil$, respectively. 
Let $j_1$ and $j_2$ be two distinct integers in
$[n']\setminus S \cup \{j\}$. 
The existence of such integers is guaranteed by the bound on $r$.
We now express $\tau$ as a composition of 
$(r/2)$-functions and $2$-functions, and then use induction hypothesis to complete
the proof. The decomposition of $\tau$ is motivated by the following
identity:
\begin{align*}	
  a_j+\prod_{s \in S}g(a_s)  = &\;\; a_j + \left(a_{j_1}+\prod_{s \in A}g(a_s)-a_{j_1}\right) \left(a_{j_2}+\prod_{s \in B}g(a_s)-a_{j_2}\right) \\
   = &\;\; a_j+ \left(a_{j_1}+\prod_{s \in A}g(a_s)\right) \left(a_{j_2}+\prod_{s \in B}g(a_s)\right)  \\
  &\; - \left( a_{j_1}+\prod_{s \in A}g(a_s)\right) a_{j_2} - a_{j_1} \left(a_{j_2}+\prod_{s \in B}g(a_s)\right) +a_{j_1}  a_{j_2}.
\end{align*}

Therefore, we consider three permutations $\gamma$, $\tau_{A}$ and $\tau_{B}$ such that for any $\<a_1,\ldots,a_{n'}\> \in \Z_m^{n'}$ their maps are given as follows: 
\begin{align*}
\gamma \left(\<a_1,\ldots,a_n\>\right) & = \<a_1,\ldots , a_{j-1},a_{j}+a_{j_1} a_{j_2},a_{j+1},\ldots ,a_{n'}\>, \\
\tau_A\left(\<a_1,\ldots,a_{n'}\>\right) & = \<a_1,\ldots , a_{j_1-1},a_{j_1}+\textstyle\prod_{s \in A} g(a_{s}),a_{j_1+1},\ldots ,a_{n'}\>, \text{ and } \\
\tau_B \left(\<a_1,\ldots,a_{n'}\>\right) & = \<a_1,\ldots , a_{j_2-1},a_{j_2}+\textstyle\prod_{s \in B} g(a_{s}),a_{j_2+1},\ldots ,a_{n'}\>, 
\end{align*}
where both the multiplications and additions are in $\Z_m$. 
Using the identity it is easy to verify the following decomposition of $\tau$: 
\[\tau =\tau_{B}^{-1} \circ \gamma^{-1} \circ \tau_{A}^{-1} \circ \gamma \circ \tau_{B} \circ \gamma^{-1} \circ \tau_{A} \circ \gamma .\]
Clearly, $\gamma$ is
a $2$-function, while $\tau_A$ and $\tau_B$ are $\lfloor r/2 \rfloor$-function and $\lceil r/2 \rceil$-function, respectively, and belong to $\mathcal{S}$.
By induction hypothesis $\tau_A$ and $\tau_B$ can be expressed
as a composition of $2$-functions.
Thus their inverses too. Hence we obtain a
decomposition of $\tau$ in terms of $2$-functions. The bound on $k_r$, the length of
the decomposition, follows from the following recurrence:
\[T(r) \leq 2T(\lfloor r/2 \rfloor) + 2 T(\lceil r/2 \rceil) + 4.\]
\end{proof}
We would like to mention that another decomposition of $\tau$ in terms of $2$-functions can be obtained by following the proof of~\cite{CG75}, albeit with a much worse bound on the value of $k_r$.
Further, by strengthening the induction hypothesis, it is easily seen that the above proof can be generalized to hold for certain special type of $r$-functions. Let $\beta$
be an $r$-function, $r \leq n'-3$, such that for any $\<a_1,\ldots,a_{n'}\> \in \Z_m^{n'}$,
\[\beta \left(\<a_1,\ldots,a_{n'}\>\right) = \<a_1,\ldots , a_{i-1},a_i+f_e(a_{i_1},\ldots,a_{i_r}),a_{i+1},\ldots ,a_{n'}\>, \]
where the function $f_e : \Z_m^{r} \to \Z_m$ is defined by: 
\[f_e(x)=
\begin{cases}
b & \text{if } x=e, \\
0 & \text{otherwise},
\end{cases}\]
for some $b \in \Z_m$ and $e\in \Z_m^{r}$, i.e., $f_e$ is some constant multiple of the characteristic function of the tuple $e$. A crucial step in the proof is to express $f_e$ as a product of indicator functions. In this case we consider the following functions $g_{i_j} : \Z_m \to \Z_m$ for $1 \leq j \leq r$.
Define $g_{i_1} : \Z_m \to \Z_m$ such that for any $y \in \Z_m$, $g_{i_1}(y)=b$ if $y=e_1$, and $0$ otherwise. For any $2\le j \le r$, define $g_{i_j} : \Z_m \to \Z_m$ as $g_{i_j}(y)=1$ if $y=e_i$, and $0$ otherwise. 
By definition, we have $f_e(x_1,\ldots,x_{r})=g_{i_1}(x_1)g_{i_2}(x_2)\cdots g_{i_r}(x_{r})$. Thus we get the following generalization of Lemma~\ref{lem:catalytic-decomposition}.
\begin{lemma}
\label{lem:general-catalytic-decomposition}
For any $m \in \N$ and $1 \le r \le n'-3$, let $\tau$ be an $r$-function such that for any $\<a_1,\ldots,a_{n'}\> \in \Z_m^{n'}$,
\[\tau \left(\<a_1,\ldots,a_{n'}\> \right) = \<a_1,\ldots , a_{j-1},a_j+f_e(a_{i_1},\ldots,a_{i_r}),a_{j+1},\ldots ,a_{n'}\>\]
where the function $f_e : \Z_m^{r} \to \Z_m$ is defined by: $f_e(x)=b$ if $x=e$; and $0$ otherwise. Then there exists a set of $2$-functions $\tau_{1},\ldots,\tau_{k_r}$ such that $\tau=\tau_{k_r}\circ \cdots \circ \tau_{1}$, and $k_r\le 4r^2 - 3$.
\end{lemma}

\subparagraph{Comment on decomposition of general $r$-functions for $3\le r \le n'-3$:}
Any function $f : \Z_m^{r} \to \Z_m$ can be expressed as $f=\sum_e c_e\cdot \chi_e$, where $\chi_e$ is the characteristic function of the tuple $e\in \Z_m^{r}$, and $c_e \in \Z_m$. Thus Lemma~\ref{lem:general-catalytic-decomposition} suffices to argue that any $r$-function can be decomposed into a set of $2$-functions. However, the implied bound on $k_r$ may not be small. In particular, the number of tuples where $f$ takes
non-zero value might be large. 

It remains to decompose $\alpha_{n'-1}$ and $\alpha_{n'}$. The following lemma about
cycles that intersect at exactly one point serves as a key tool in our decomposition.

\begin{lemma}
\label{lem:cycle-isolation}
Suppose there are two cycles, $\sigma=(t, a_1, \cdots, a_{\ell -1})$ and $\tau=(t, b_1, \cdots, b_{\ell -1})$, of length $\ell \geq 2$ such that $a_i \ne b_j$ for every $i,j \in [\ell -1]$. Then, $(\sigma \circ \tau)^\ell \circ (\tau \circ \sigma)^\ell  = \sigma^2$.
\end{lemma}
\begin{proof}
By Proposition~\ref{prop:intersect-cycle}, we have 
\begin{align*}
\beta:=\tau \circ \sigma & =(t, a_1, \cdots, a_{\ell -1}, b_1, \cdots, b_{\ell -1}), \text{ and } \\ 
\gamma:=\sigma \circ \tau & =(t, b_1, \cdots, b_{\ell -1}, a_1, \cdots, a_{\ell -1}).
\end{align*}
Both $\beta$ and $\gamma$ are cycles of length $2\ell- 1$. Also note that $2\ell-1$ is co-prime with $\ell $. Thus both $\beta^\ell $ and $\gamma^\ell $ are also cycles of length $2\ell - 1$ given as follow: 
\begin{align*}
  \beta^\ell & = (t, b_1, a_1, b_2, a_2, \cdots, b_{\ell -1}, a_{\ell -1}), \text{ and }\\
  \gamma^\ell & = (t, a_1,b_1, a_2, b_2, \cdots, a_{\ell -1}, b_{\ell -1}).
\end{align*}
Now by Proposition~\ref{prop:cycle-conjugate},
\begin{align*}
\sigma\circ \beta^\ell  \circ \sigma^{-1} & = (\sigma(t), \sigma(b_1), \sigma(a_1), \cdots,\sigma(a_{\ell -2}), \sigma(b_{\ell -1}), \sigma(a_{\ell -1})) \\
&=(a_1, b_1, a_2, b_2, \cdots, a_{\ell -1}, b_{\ell -1},t)=\gamma^\ell .
\end{align*}
Therefore, 
\begin{align*}
\beta^\ell  \circ \sigma^{-1} &= (t, a_{\ell -1}, a_{\ell -2},\cdots, a_1) \circ (t, a_1,b_1, a_2, b_2, \cdots, a_{\ell -1}, b_{\ell -1}) \\
&=(a_1,b_1)(a_2,b_2)\cdots(a_{\ell -1},b_{\ell -1}).
\end{align*}
It is thus evident that  $\left(\beta^\ell  \circ \sigma^{-1}\right)^2$ is the identity permutation. Hence,
\begin{align*}
\gamma^\ell  \circ \beta^\ell  &=\sigma \circ \beta^\ell  \circ \sigma^{-1}\circ \beta^\ell \\
&=\sigma \circ \beta^\ell  \circ \sigma^{-1}\circ \beta^\ell  \circ \sigma^{-1}\circ \sigma\\
&=\sigma^2.
\end{align*}
\end{proof}

Before going into the detailed description of the decomposition procedure, let us briefly discuss the main idea. Here we first consider $\alpha_{n'}$. The case of $\alpha_{n'-1}$ will be analogous. Recall that $\alpha_{n'}=(\<00\cdots 00\>, \<00\cdots 01\>, \<00\cdots 02\>, \cdots, \<00\cdots 0(m-1)\>)$ is a cycle of length $m$. For $a=(m+1)/2$, we define $\sigma=(\<00\cdots 0(0\cdot a)\>, \<00\cdots 0(1\cdot a)\>, \<00\cdots 0(2\cdot a)\>, \cdots, \<00\cdots 0((m-1)\cdot a)\>)$, and  $\tau=(\<(0\cdot a)00\cdots 0\>, \<(1\cdot a)00\cdots 0\>, \<(2\cdot a)00\cdots 0\>, \cdots, \<((m-1)\cdot a)00\cdots 0\>)$, where the multiplication is in $\Z_m$.  Since $m$ is co-prime with $(m+1)/2$, $\sigma$ and $\tau$ are cycles of length $m$. (Here we use the fact that $m$ is odd.) Observe that $\sigma^2=\alpha_{n'}$, so by applying Lemma~\ref{lem:cycle-isolation} to $\sigma$ and $\tau$ we get $\alpha_{n'}$.  It might seem we didn't make much progress towards decomposition, as now instead of one $(n'-1)$-function $\alpha_{n'}$ we have to decompose two $(n'-1)$-functions $\sigma$ and $\tau$. However, we will not decompose $\sigma$ and $\tau$ directly, but rather we obtain a decomposition for $(\sigma \circ \tau)^m$ and $(\tau \circ \sigma)^m$. Surprisingly this can be done using Lemma~\ref{lem:general-catalytic-decomposition} although indirectly. 

We consider an $(n'-3)$-function $\sigma'$  whose cycle decomposition contains $\sigma$ as one of its cycles. Similarly we consider a $3$-function $\tau'$ whose cycle decomposition contains $\tau$ as one of its cycles. We carefully choose these $\sigma'$ and $\tau'$ such that $(\sigma' \circ \tau')^m=(\sigma \circ \tau)^m$ and $(\tau' \circ \sigma')^m=(\tau \circ \sigma)^m$. We will use Lemma~\ref{lem:general-catalytic-decomposition} to directly decompose $\sigma'$ and $\tau'$ to get the desired decomposition.

\begin{lemma}
\label{lem:isolating-decomposition}
For any $n'-1\le i \le n'$, let $\alpha_i$ be the permutation over $\Z_m^{n'}$ given by Definition~\ref{def:alpha-i} where $m$ is odd. Then, there exists a set of $2$-functions $\alpha_{i,1},\ldots,\alpha_{i,k_i}$ such that $\alpha_i=\alpha_{i,k_i}\circ \cdots \circ \alpha_{i,1}$, and $k_i = O(m \cdot (i-1)^2)$.
\end{lemma}
\begin{proof}
  For the sake of brevity, we will only describe the procedure to decompose $\alpha_{n'}$ into a set of $2$-functions. The decomposition of $\alpha_{n'-1}$ is completely analogous. (We comment on this more at the end of the proof.) 

  Let $\sigma$ be the following permutation: for any $\<a_1,\ldots,a_{n'}\> \in \Z_m^{n'}$,
\[\sigma \left(\<a_1,\ldots,a_{n'}\>\right)=\<a_1,\ldots,a_{n'-1},a_{n'}+f(a_{1},\ldots,a_{n'-1})\>\]
where the function $f : \Z_m^{n'-1} \to \Z_m$ is defined as follows: 
\[ f(x)=
\begin{cases}
(m+1)/2 & \text{if } x=\<0,\ldots ,0\>, \\
0 & \text{otherwise}.
\end{cases}
\]
Note that $(m+1)/2$ is well defined because $m$ is odd. Further, since $m$ and $(m+1)/2$ are co-prime, $\sigma$ is a $m$ length cycle. Moreover,  $\sigma^2=\alpha_{n'}$.
The description of $\sigma$ uses crucially that $m$ is odd. If $m$ were not odd, then the description fails.  Indeed finding a substitute for $\sigma$ is the main hurdle that needs to be addressed to handle the case when $m$ is even. 

We also consider another permutation $\tau$ such that for any $\<a_1,\ldots,a_{n'}\> \in \Z_m^{n'}$,
\[\tau\left(\<a_1,\ldots,a_{n'}\>\right)=\<a_1+f(a_{2},\ldots,a_{n'}),a_2,\ldots , a_{n'}\>\]
where the function $f$ is the same as in the definition of $\sigma$.
So $\tau$ is also a cycle of length $m$. 
Let $t=(m+1)/2$. Then the cycle decomposition of $\sigma$ and $\tau$ are, 
\begin{align*}
\sigma &= \left(\<0,0,\ldots, 0,0\cdot t\>, \<0,0,\ldots ,0,1\cdot t\>, \<0,0,\ldots ,0,2\cdot t\>, \cdots, \<0,0,\ldots ,0,(m-1)\cdot t\> \right), \text{ and } \\
\tau &= \left(\<0\cdot t, 0,0,\ldots, 0\>, \<1\cdot t ,0,0, \ldots, 0\>, \<2\cdot t, 0, 0, \ldots, 0\>, \cdots, \<(m-1)\cdot t,0,0,\ldots, 0\> \right) 
\end{align*}
where the multiplication is in $\Z_m$.
Observe that $\<0,\ldots,0\>$ is the only common element involved in both cycles $\sigma$ and $\tau$. Therefore, by Lemma~\ref{lem:cycle-isolation},
\[(\sigma \circ \tau)^m\circ (\tau \circ \sigma)^m = \sigma^2=\alpha_{n'}.\]
We note that both $\sigma$ and $\tau$ are $(n'-1)$-functions. Thus so far it is not clear how the above identity helps us to decompose $\alpha_{n'}$. We now define two
more permutations $\sigma'$ and $\tau'$ such that they are themselves decomposable into
$2$-functions and, moreover, $(\sigma' \circ \tau')^m=(\sigma \circ \tau)^m$ and $(\tau' \circ \sigma')^m=(\tau \circ \sigma)^m$.

The permutations $\sigma'$ and $\tau'$ are defined as follows: for any $\<a_1,\ldots,a_{n'}\> \in \Z_m^{n'}$,
\begin{align*}
  \sigma'\left(\<a_1,\ldots,a_{n'}\>\right) &=\<a_1,\ldots , a_{n'-1},a_{n'}+g(a_{1},\ldots,a_{n'-3})\>, \text{ and } \\
  \tau'\left(\<a_1,\ldots,a_{n'}\>\right) &= \<a_1+h(a_{n'-2},a_{n'-1},a_{n'}),a_2,\ldots , a_{n'}\>
  \end{align*}
where $g : \Z_m^{n'-3} \to \Z_m$ is the following map: 
\[g(x)=
\begin{cases}
(m+1)/2 & \text{if } x=\<0,\ldots ,0\>, \\
0  & \text{otherwise}, 
\end{cases}
\]
and $h : \Z_m^{3} \to \Z_m$ is similarly defined but on lower dimension: 
\[h(x)=
\begin{cases}
(m+1)/2 & \text{if } x=\<0,0 ,0\>, \\
0 & \text{otherwise}.
\end{cases}
\]
Since $m$ and $(m+1)/2$ are co-prime, $\sigma'$ is composed of $m^2$ disjoint cycles, each of length $m$. Similarly, $\tau'$ is composed of $m^{(n'-4)}$ disjoint cycles, each of length $m$. Moreover, $\sigma$ is one of the cycles among $m^2$ disjoint cycles of $\sigma'$ and $\tau$ is one of the cycles among $m^{(n'-4)}$ disjoint cycles of $\tau'$. So, we can write
\begin{align*}
\sigma'&=\sigma\circ C_1\circ \cdots \circ C_{m^2-1} , \text{ and }\\
\tau'&=\tau\circ C'_1\circ \cdots \circ C'_{m^{(n'-4)}-1}
\end{align*}
where each of $C_i$'s and $C'_j$'s is a $m$ length cycle.
An important fact regarding $\sigma'$ and $\tau'$ is that the \emph{only} element
moved by both is the all-zeros tuple $\<0,\ldots,0\>$. This is easily seen from
their definitions. Recall we had observed that the all-zeroes is, in fact, moved by
$\sigma$ and $\tau$. In other words, cycles in the set \(\{C_1,\ldots, C_{m^2-1},C'_1, \ldots, C'_{m^{(n'-4)}-1} \}\) are mutually disjoint, as well as they are disjoint from
$\sigma$ and $\tau$.
Thus, using the fact that $C_i$'s and $C'_j$'s are $m$ length cycles, we have
\begin{align*}
(\tau'\circ \sigma')^m & = (\tau\circ \sigma)^m, \text{ and } \\
(\sigma'\circ \tau')^m & = (\sigma \circ \tau)^m.
\end{align*}
Therefore, we can express $\alpha_{n'}$ in terms of $\sigma'$ and $\tau'$, 
\[\alpha_{n'}=\sigma^2=(\sigma \circ \tau)^m\circ (\tau \circ \sigma)^m = (\sigma' \circ \tau')^m\circ (\tau' \circ \sigma')^m. \]
But, by definition, $\sigma'$ and $\tau'$ are an $(n'-3)$-function and a $3$-function, respectively. Furthermore, they satisfy the requirement of Lemma~\ref{lem:general-catalytic-decomposition}. Hence both $\sigma'$ and $\tau'$ can be decomposed into a set of $2$-functions. As a result we obtain a decomposition of $\alpha_{n'}$ into a set of $k_{n'}$ many $2$-functions, where $k_{n'} \le 2m \left(4(n'-3)^2 - 3 + 4 \cdot 3^2 - 3\right) = m\left(8(n'-3)^2 + 60\right) \le 60 \cdot m \left(n'-3\right)^2$.

The permutation $\alpha_{n'-1}$ can be decomposed in a similar fashion. Note that $\alpha_{n'-1}$ is composed of $m$ disjoint $m$ length cycles. Each of the $m$ disjoint cycles can be decomposed using the procedure described above. If we do so,  we get the length $k_{n'-1} = O( m^2 \cdot (n'-3)^2)$, which would suffice for our purpose. 
However, we can improve this bound to $O(m \cdot (n'-2)^2)$ by a slight modification of $\sigma,\tau,\sigma'$ and $\tau'$. Below we define these permutations and leave the rest of the proof to the reader as it is analogous to the proof above. The argument should be carried with $\sigma,\tau,\sigma'$ and $\tau'$ defined as follows: for any $\<a_1,\ldots,a_{n'}\> \in \Z_m^{n'}$,
\begin{align*}
\sigma \left(\<a_1,\ldots,a_{n'}\>\right) &=\<a_1,\ldots , a_{n'-2},a_{n'-1}+f(a_{1},\ldots,a_{n'-2}),a_{n'}\> , \\
\tau\left(\<a_1,\ldots,a_{n'}\>\right) & = \<a_1+f(a_{2},\ldots,a_{n'-1}),\ldots , a_{n'-1},a_{n'}\>,\\
\sigma'\left(\<a_1,\ldots,a_{n'}\>\right) &=\<a_1,\ldots , a_{n'-2},a_{n'-1}+g(a_{1},\ldots,a_{n'-3}),a_{n'}\>, \text{ and }\\
\tau'\left(\<a_1,\ldots,a_{n'}\>\right) &=\<a_1+h(a_{n'-2},a_{n'-1}),a_2,\ldots , a_{n'-1},a_{n'}\>, \\
\end{align*}
where the function $f : \Z_m^{n'-2} \to \Z_m$ is defined as: $f(x)=(m+1)/2$ if $x=\<0,\ldots ,0\>$; otherwise $f(x)=0$, the function $g : \Z_m^{n'-3} \to \Z_m$ is defined as: $g(x)=(m+1)/2$ if $x=\<0,\ldots ,0\>$; otherwise $g(x)=0$ and the function $h : \Z_m^{2} \to \Z_m$ is defined as: $h(x)=(m+1)/2$ if $x=\<0 ,0\>$; otherwise $h(x)=0$. The only difference is now that both $\sigma$ and $\tau$ are composed of $m$ disjoint cycles of length $m$, instead of just one cycle as in case of $\alpha_{n'}$. However every cycle of $\sigma$ has non-empty intersection with exactly one cycle of $\tau$, and furthermore, the intersection is singleton. Hence we can still apply Lemma~\ref{lem:cycle-isolation} with $m$ pairs of $m$ length cycles, where each pair consists of one cycle from $\sigma$ and another from $\tau$ such that they have non-empty intersection.
\end{proof}

\subsubsection*{Step 3: Invocation of Lemma~\ref{lem:cycle-counter}}

To finish the construction we just replace $\alpha_i$'s in $\alpha=\alpha_{n'}\circ \cdots \circ \alpha_{1}$ by $\alpha_{i,k_i}\circ \cdots \circ \alpha_{i,1}$. Take the resulting sequence as $\sigma=\sigma_k\circ \cdots \circ \sigma_1$, where $k=\sum_{i=1}^{n'}k_i \le c_3\cdot m \cdot n'^3$ for some constant $c_3>0$. Now we apply Lemma~\ref{lem:cycle-counter} to get a space-optimal counter over $\Z_m^n$. In Lemma~\ref{lem:cycle-counter}, we set $k'=m^{r'}$ such that $r'$ is the smallest integer for which $m^{r'} \ge k$ and set $\mathcal{D'}=\Z_m^{r'}$. Hence $r'\le \log_mk +1 \le 3\log_m n + c$, for some constant $c > 1$. Since each of $\sigma_i$'s is a $2$-function, by Claim~\ref{obs:r-function-DAT} it can be implemented using a decision assignment tree $T_i$ such that $\Read(T_i)=3$ and $\Write(T_i)=1$. In Lemma~\ref{lem:cycle-counter}, we use the standard space-optimal Gray code over $\mathcal{D'}=\Z_m^{r'}$ as $C'$. The code $C'$ can be implemented using a decision assignment tree $T'$ with $\Read(T')=r'$ and $\Write(T')=1$ (Theorem~\ref{thm:general-Graycode}). Hence the final counter implied from Lemma~\ref{lem:cycle-counter} can be computed using a decision assignment tree $T$ with $\Read(T)\le 4\log_m n $ and $\Write(T)=2$.

\paragraph*{Odd permutations: the real bottleneck}
We saw our decomposition procedure over $\Z_m^{n}$ into $2$-functions (Lemma~\ref{lem:isolating-decomposition}) requires that $m$ be odd.
It is natural to wonder whether such a decomposition exists irrespective of the parity of $m$. That is, we would like to design a set of $2$-functions that generates a cycle of length $m^{n}$ where $m$ is even. Unfortunately, the answer to this question
is, unequivocally, \textsf{NO}. A cycle of length $m^{n}$ is an
\emph{odd permutation} when $m$ is even. However,  Coppersmith and Grossman~\cite{CG75} showed that $2$-functions are \emph{even permutations} for $m=2$. Their proof
is easily seen to be extended to hold for \emph{all} $m$. Hence, $2$-functions can
only generate even permutations. Thus, ruling out the possibility of decomposing
a full cycle into $2$-functions when $m$ is even. In fact, Coppersmith and Grossman~\cite{CG75} showed that $(n-2)$-functions are all even permutations. Therefore, the decomposition remains impossible even if $(n-2)$-functions are allowed. 

We would like to give further evidence that odd permutations are indeed the \emph{bottleneck}
towards obtaining an efficient decision assignment tree implementing a space-optimal
counter over $\Z_2^{n}$.
Suppose $T$ is a decision assignment tree that implements an odd permutation
$\sigma$ over $\Z_2^{n}$ such that $\Read(T) = r$ and $\Write(T) = w$.
From the basic group theory we know that there exists
an even permutation $\tau$ such that $\tau \circ \sigma$ is a cycle of length $2^{n}$. Following the argument in~\cite{CG75}, but using an efficient decomposition of
$(n-3)$-functions into $2$-functions (in particular, Lemma~\ref{lem:general-catalytic-decomposition}), we see that any even permutation can be
decomposed into a set of $2$-functions of size at most $O(2^{n-1}n^{2})$. Indeed,
any even permutation can be expressed as a composition of at most $2^{n-1}$ $3$-cycles, while any $3$-cycle can be decomposed into at most $O(n^2)$ $2$-functions.
We implement $\tau\circ\sigma$ using Lemma~\ref{lem:cycle-counter} to obtain a
space-optimal counter over $\Z_2^{2n+O(\log n)}$ that is computed by a decision
assignment tree with read complexity $n + O(\log n) + r$ and write complexity $w+1$. Thus, savings made on implementing an odd permutation translate directly to the savings possible in the implementation of some space-optimal counter. By the argument in~\cite{Raskin17} we know that $r \in \Omega(n)$. 

Recall the main idea behind $2$-Gray codes in Section~\ref{sec:binarycounter}, when $m=2$, is to generate the cycle of length $2^{n}-1$ efficiently. Since a cycle of length $2^{n}-1$ is an even permutation, we could decompose it into $2$-functions. In particular, we could choose an appropriate $3$-cycle as $\alpha_{n}$ in  Claim~\ref{clm:full-cycle} such that $\alpha$ becomes a cycle of length $2^{n}-1$. However, the counter thus obtained misses out on $O(n^3)$ strings, instead of $O(n)$. 

\subsection{Getting counters for even $m$}
\label{sec:general-counter}
We can combine the results from Theorem~\ref{thm:binary-counter} and Theorem~\ref{thm:odd-counter} to get a counter over $\Z_m^n$ for any even $m$. We have already mentioned in Section~\ref{sec:chinese} that if we have space-optimal quasi-Gray codes over the alphabet $\Z_2$ and $\Z_{\om}$ with $\om$ being odd then we can get a space-optimal quasi-Gray code over the alphabet $\Z_m$ for any even $m$. Unfortunately in Section~\ref{sec:binarycounter}, instead of space-optimal counters we were only able to generate a counter over binary alphabet that misses $O(n)$ many strings. As a consequence we cannot directly apply Theorem~\ref{thm:chinese}. The problem is following. Suppose $m=2^\ell \om$ for some $\ell > 0$ and odd $\om$. By the argument used in the proof of Lemma~\ref{clm:stitch-binary} we know that there is a counter over $(\Z_{2^\ell})^{n-1}$ of the same length as that over $(\Z_2)^{\ell(n-1)}$. Furthermore the length of the counter is of the form $2^{O(\log n)}(2^{n'}-1)$, for some $n'$ that depends on the value of $\ell n$ (see the construction in Section~\ref{sec:binarycounter}). Now to apply Theorem~\ref{thm:chinese} as in the proof of Lemma~\ref{lem:stitch-counter}, $2^{n'}-1$ must be co-prime with $\om$. However that may not always be the case. Nevertheless, we can choose the parameters in the construction given in Section~\ref{sec:binarycounter} to make $n'$ such that $2^{n'}-1$ is co-prime with $\om$. This is always possible because of the following simple algebraic fact. In the following proposition we use the notation $\Z_{\om}^{*}$ to denote the multiplicative group modulo $\om$ and $\ord_{\om}(e)$ to denote the order of any element $e \in \Z_{\om}^{*}$.
\begin{proposition}
\label{prop:co-prime}
For any $n \in \N$ and odd $\om \in \N$, consider the set $\mathcal{S}=\{n,n+1,\cdots,n+\ord_{\om}(2)-1\}$. Then there always exists an element $n'\in \mathcal{S}$ such that $2^{n'}-1$ is co-prime to $\om$.
\end{proposition}
\begin{proof}
Inside $\Z_{\om}^*$, consider the cyclic subgroup $G$ generated by $2$, i.e., $G=\{1,2,\cdots,2^{\ord_{\om}(2)}\}$. Clearly, $\{2^{s} \pmod \om \mid s \in \mathcal{S}\}=G$. Hence there exists an element $n'\in \mathcal{S}$ such that $2^{n'}-1 \pmod \om = 1$. It is clear that if $\gcd(2^{n'}-1,\om)=1$ then we are done. Now the proposition follows from the following easily verifiable fact: for any $a,b,c \in \N$, if $a\equiv b \pmod c$ then $\gcd(a,c)=\gcd(b,c)$.
\end{proof}
So for any $n \in \N$, in the proof of Lemma~\ref{lem:stitch-counter} we take the first coordinate to be $\Z_m^i$ for some suitably chosen $i\ge 1$ instead of just $\Z_m$. The choice of $i$ will be such that the length of the counter over $(\Z_{2^\ell})^{n-i}$ will become co-prime with $\om$. The above proposition guarantees the existence of such an $i \in [\ord_{\om}(2)]$. Hence we can conclude the following.
\begin{theorem}
\label{thm:general-counter}
For any even $m \in \N$ so that $m=2^{\ell} \om$ where $\om$ is odd, there is a quasi-Gray code $C$ over $\Z_m^n$ of length at least $m^n-O(n\om^n)$, that can be implemented using a decision assignment tree which reads at most $O(\log_m n+\ord_{\om}(2))$ coordinates and writes at most $3$ coordinates.
\end{theorem}

\begin{remark}
So far we have only talked about implementing $next(\cdot,\cdot)$ function for any counter. However we can similarly define complexity of computing $prev(\cdot,\cdot)$ function in the decision assignment tree model. We would like to emphasize that all our techniques can also be carried over to do that. To extend the result of Section~\ref{sec:binarycounter}, we take the inverse of the linear transformation matrix $A$ and follow exactly the same proof to get the implementation of $prev(\cdot,\cdot)$. As a consequence we achieve exactly the same bound on read and write complexity. Now for the quasi-Gray code over $\Z_m$ for any odd $m$, instead of $\alpha$ in the Step $1$ (in Section~\ref{sec:oddcounter}), we simply consider the $\alpha^{-1}$ which is equal to $\alpha_1^{-1}\circ \alpha_2^{-1}\circ \cdots \circ \alpha_{n'}^{-1}$. Then we follow an analogous technique to decompose $\alpha_i^{-1}$'s and finally invoke Lemma~\ref{lem:cycle-counter}. Thus we get the same bound on read and write complexity.
\end{remark}

\section{Lower Bound Results for Hierarchical Counter}
\label{sec:LB-hierarchy}
We have mentioned that the best known lower bound on the depth of a decision assignment tree implementing any space-optimal quasi-Gray code is $\Omega(\log_m n)$. Suppose we are interested in getting a space-optimal quasi-Gray code over the domain $\Z_{m_1}\times \cdots \times \Z_{m_n}$.
A natural approach to achieve the $O(\log n)$ bound
is to build a decision assignment tree with an additional restriction that a
variable labels \emph{exactly} one internal node. 
We call any counter associated with such a special type of decision assignment tree by \emph{hierarchical counter}. We note that the aforementioned
restriction doesn't guarantee an $O(\log n)$ bound on the depth; for example,
a decision assignment tree implementing a counter that increments by 1
to obtain the next element, reads a variable exactly once. Nevertheless, informally,
 one can argue that if a hierarchical counter is lopsided it must mimic the decision assignment tree in the previous example.  
We show that a space-optimal hierarchical counter over strings of length $3$ exists if and only if $m_2$ and $m_3$ are co-prime.

\begin{theorem}
\label{thm:hierarchy-LB}
Let $m_1 \in \{2,3\}$ and $m_2 , m_3 \geq 2$. There is no space-optimal hierarchical counter over $\Z_{m_1} \times \Z_{m_2} \times \Z_{m_3}$ unless $m_2$ and $m_3$ are co-prime. 
\end{theorem}
In fact, we conjecture that the above theorem holds for all $m_1 \geq 2$.
\begin{conjecture}
There is no space-optimal hierarchical counter over $\Z_{m_1} \times \Z_{m_2} \times \Z_{m_3}$ unless $m_2$ and $m_3$ are co-prime. 
\end{conjecture}
It follows already from Theorem~\ref{thm:chinese} that a space-optimal
hierarchical counter exists when $m_2$ and $m_3$ are co-prime. Hence Theorem~\ref{thm:hierarchy-LB} establishes the necessity of co-primality in Theorem~\ref{thm:chinese}. It also provides a theoretical explanation behind the already known fact, derived from an exhaustive search in~\cite{BGPS14}, that there exists no space-optimal counter over $(\Z_2)^3$ that reads at most $2$ bits to generate the next element. 
\begin{proof}[Proof of Theorem~\ref{thm:hierarchy-LB}]
  We first establish the proof when $m_1=2$, and then reduce the case when
  $m_1=3$ to $m_1=2$.

  Following the definition of hierarchical counter, without loss of generality, we assume that it reads the second coordinate if the value of the first coordinate is $0$; otherwise it reads the third coordinate. Let $C=(\<a_0,b_0,c_0\>, \ldots, \<a_{k-1},b_{k-1},c_{k-1}\>)$ be the cyclic sequence associated with a space-optimal hierarchical counter over $\Z_{2} \times \Z_{m_2} \times \Z_{m_3}$. We claim that for any two consecutive elements in the sequence, $\<a_i,b_i,c_i\>$ and $\<a_{(i+1) \bmod k},b_{(i+1) \bmod k},c_{(i+1) \bmod k}\>$, the following must hold: 
  \begin{enumerate}
  \item If $a_i=0$, then $b_{i+1} \neq b_i$. For the sake of contradiction suppose not, i.e, $b_{i+1} = b_i$. Since $a_i=0$ we do not read the third coordinate, thus, implying $c_{i+1}=c_i$. Hence, $a_i \neq a_{i+1} = 1$. Therefore, we observe that the second coordinate never changes henceforth. Thence it must be that 
    $b_1=b_2=\cdots = b_k$. We thus contradict space-optimality.
    Similarly, we can argue if $a_i=1$, then $c_{i+1} \neq c_i$.
  \item For $i \in \{0,1,\ldots,k-1\}$, $a_{(i+1)\bmod k} \neq a_i$. Assume otherwise. Let $i \in \{0\}\cup[k-1]$ be such that $a_{(i-1)\bmod k} \neq a_i = a_{(i+1)\bmod k}$. Clearly, such an $i$ exists, else $a_1=a_2=\cdots =a_k$ which contradicts space-optimality. We are now going to argue that there exist two distinct elements in the domain $\Z_{2} \times \Z_{m_2} \times \Z_{m_3}$ such that they have the   same successor in $C$, thus contradicting that $C$ is generated by a space-optimal counter. For the sake of simplicity let us assume $a_i = 0$. The other case
    is similar. Since $a_i=0$, we have $a_{i-1}=1$ and $a_{i+1} = 0$. Thus, we
    have the tuple $\<0,b_{i+1},c_{i+1}\>$ following the tuple $\<0,b_i,c_i\>$ in $C$. We claim that the successor of the tuple $\<1,b_{i+1},c_{i-1}\>$ is also
     $\<0,b_{i+1},c_{i+1}\>$. It is easily seen because $c_{i+1} = c_i$.  
  \end{enumerate}

  The second item above says that the period of each element in the first coordinate is 2. That is, in the sequence $C$ the value in the first coordinate repeats every second step. We will now argue that the value in the second coordinate (respectively, third coordinate)
  repeats every $2m_2$ steps (respectively, $2m_3$ steps). We will only establish the claim for the second coordinate, as the argument for the third coordinate is similar. First of all observe that $C$ induces a cycle on the elements of $\Z_{m_2}$. In particular, Consider the sequence $b_0, b_2, b_4, \ldots , b_\ell$
  such that $b_\ell$ is the first time when an element in the sequence
  has appeared twice. First of all, for $C$ to be a valid cycle $b_\ell$ must be equal to $b_0$. Secondly, $b_0, b_2, \ldots, b_{\ell-2}$ are the only distinct elements that appear in the second coordinate throughout $C$. This is because
  the change in the second coordinate depends, deterministically, on what value it holds prior to the change. Thus, for $C$ to be space-optimal $\ell/2$ must be equal to $m_2$. Hence, an element in the second coordinate repeats every $2m_2$ steps in $C$.
  
  Therefore, for the cyclic sequence $C$, we get that $k=\lcm(2m_2,2m_3)=2m_2m_3/\gcd(m_2,m_3)$. Hence if $m_2$ and $m_3$ are not co-prime then $k<2m_2m_3$ implying $C$ to be not space-optimal. 

  We now proceed to  the case $m_1=3$.  As mentioned in the beginning, we will show that from any space-optimal hierarchical counter over $\Z_{3} \times \Z_{m_2} \times \Z_{m_3}$ we can obtain a space-optimal hierarchical counter over $\Z_{2} \times \Z_{m_2} \times \Z_{m_3}$.

  Let $C=(\<a_0,b_0,c_0\>, \ldots, \<a_{k-1},b_{k-1},c_{k-1}\>)$ be the cyclic sequence associated with a space-optimal hierarchical counter $T$ over $\Z_{3} \times \Z_{m_2} \times \Z_{m_3}$. We query the first coordinate, say $x_1$, at the root of $T$.
  Let $T_0$, $T_1$, and $T_2$ be the three subtrees rooted at $x_1 = 0$, $1$, and $2$, respectively. Without loss of generality, we assume that the roots of
  $T_0$ and $T_2$ are labeled by the second coordinate, say $x_2$, and the
  root of $T_1$ is labeled by the third coordinate, say $x_3$. We say that
  there is a \emph{transition} from $T_i$ to $T_j$ if there exists a leaf in $T_i$ labeled with an instruction to change $x_1$ to $j$. We will use the shorthand $T_i \to T_j$ to denote this fact. For $C$ to be space-optimal there must be a transition from $T_1$ to either $T_0$ or $T_2$ or both. 
  We now observe some structural properties of $T$.  We first assume that there is a transition from $T_1$ to $T_0$. (The other case when there is a transition from $T_0$ to $T_2$ will be dual to this.)    
  \begin{enumerate}
  \item[\textit{(i)}] If $T_1 \to T_0$, then $T_2 \not\to T_0$. Suppose not, and let $\<2,x_2,\ast\>$ goes to $\<0,x'_2,\ast\>$. Since $T_1 \to T_0$, we have
    $\<1,\ast , x_3\>$ goes to $\<0, \ast ,x'_3\>$. It is easily seen that
    both $\<2,x_2,x'_3\>$ and $\<1,x'_2,x_3\>$ have the same successor in $C$ which is $\<0,x'_2, x'_3\>$. (Notice similarly we could also argue that if $T_1 \to T_2$, then $T_0 \not\to T_2$.). 
  \item[\textit{(ii)}] If $T_1 \to T_0$, then $T_1 \not\to T_2$. Suppose not, then from item \textit{(i)} we know that $T_2 \not\to T_0$ and $T_0 \not\to T_2$. Thus, in the sequence $C$ two tuples with first coordinate $0$ or $2$ is separated by a tuple with the first coordinate $1$. Since there are exactly equal number of tuples with a fixed first coordinate, $C$ can not be space-optimal.
  \item[\textit{(iii)}] If $T_1 \to T_0$, then $T_0 \not\to T_0$ and $T_1 \not\to T_1$. The argument is similar to item \textit{(i)}. 
  \end{enumerate}

  We now prune $T$ to obtain a space-optimal hierarchical counter $T'$ over $\Z_{2} \times \Z_{m_2} \times \Z_{m_3}$ as follows. We first remove the whole subtree $T_2$, i.e., only $T_0$ and $T_1$ remain incident on $x_1$. We don't make any change to $T_1$. Now we change instructions at the leaves in $T_0$ to complete
  the construction of $T'$. Let $\ell$ be a leaf in $T_0$ such that it is labeled with the assignment: $x_1 \gets 2$ and $x_2 \gets b$  for some $b \in \Z_{m_2}$. As a thought experiment, let us traverse $T_2$ in $T$ until we reach a leaf in $T_2$ with assignments of the form $x_1 \gets 1$ and $x_2 \gets b'$ for some $b' \in \Z_{m_2}$. In this case, we change the content of the leaf $\ell$ in $T_0$ in $T'$ to $x_1 \gets 1$ and $x_2 \gets b'$. The fact that $T'$ is space-optimal and hierarchical follows easily from properties \textit{(i) - (iii)} of $T$, and it being space-optimal and hierarchical. This completes the proof. 
\end{proof}

\section*{Acknowledgments}
Authors would like to thank Gerth St{\o}lting Brodal for bringing the problem to our attention, and to Petr Gregor for giving a talk on space-optimal counters in our seminar, which motivated this research. Authors thank Meena Mahajan and Venkatesh Raman for pointing out the result in~\cite{Raskin17}. We also thank anonymous reviewers for helpful suggestions that improved the presentation of the paper. 

\bibliographystyle{plainurl}
\bibliography{counter}

\end{document}